\documentclass[journal,onecolumn]{IEEEtran}

\usepackage{color}
\usepackage{graphicx}
\usepackage{amsthm}

\usepackage{subcaption}

\newtheorem{theorem}{Theorem}[section]
\newtheorem{corollary}[theorem]{Corollary}
\newtheorem{proposition}[theorem]{Proposition}
\newtheorem{lemma}[theorem]{Lemma}
\newtheorem{remark}{Remark}

\def\D{\Delta}
\def\RR{\mathbb{R}}
\def\CC{\mathbb{C}}
\def\PP{\mathbb{P}}
\def\EE{\mathbb{E}}
\def\cB{\mathcal{B}}
\def\cD{\mathcal{D}}
\def\cI{\mathcal{I}}
\def\cO{\mathcal{O}}
\def\cR{\mathcal{R}}
\def\cS{\mathcal{S}}
\def\cV{\mathcal{V}}
\def\fF{\mathfrak{F}}
\def\fI{\mathfrak{I}}
\def\fM{\mathfrak{M}}

\def\Var{\textup{\mbox{Cov}}}
\def\Fr{\textup{\mbox{Fr}}}
\def\op{\textup{\mbox{op}}}
\def\tr{\textup{\mbox{trace}}}
\def\det{\textup{\mbox{det}}}

\newcommand{\abs}[1]{\left\lvert #1 \right\rvert}

\newcommand{\norm}[1]{\left\lVert#1\right\rVert}

\usepackage{amsmath, amsfonts}

\begin{document}
%
\title{On Lipschitz Bounds of General Convolutional Neural Networks}
%
%
%

\author{Dongmian Zou,~\IEEEmembership{} \thanks{D. Zou (Email: dzou@ima.umn.edu) is with the Institute for Mathematics and its Applications, University of Minnesota; he was with Department of Mathematics, University of Maryland when writing the first draft of the paper. R. Balan (Email: rvbalan@cscamm.umd.edu) is with Department of Mathematics, University of Maryland. M. Singh (Email: Maneesh.Singh@verisk.com) is with Verisk Analytics.}
        Radu Balan,~\IEEEmembership{}
        Maneesh Singh~\IEEEmembership{}}
\maketitle

\begin{abstract}
Many convolutional neural networks (CNN's) have a feed-forward structure. In this paper, a linear program that estimates the Lipschitz bound of such CNN's is proposed. Several CNN's, including the scattering networks, the AlexNet and the GoogleNet, are studied numerically and compared to the theoretical bounds. Next, concentration inequalities of the output distribution to a stationary random input signal expressed in terms of the Lipschitz bound are established. The Lipschitz bound is further used to establish a nonlinear discriminant analysis designed to measure the separation between features of different classes. 
\end{abstract}


%
\IEEEpeerreviewmaketitle

\section{Introduction}\label{sec:intro}

\IEEEPARstart{C}{onvolutional} neural networks (CNN's) have proved to be an effective tool in various image processing tasks. The convolutional layers at different levels are capable of extracting different details from images. As a feature extractor, a CNN is stable to small variations from the input and therefore performs well in a variety of classification, detection and segmentation problems.

The scattering transform \cite{Mallat12,BM13} is a special type of CNN that can be represented with a multilayer structure (thus also called a scattering network). Although the filters are designed wavelets rather than learned, the scattering transform proves to be an effective feature extractor.  In the mathematical analysis of scattering network, it is proved \cite[Theorem 2.10]{Mallat12} that the scattering transform is invariant to translation. However, this is true only if we take the full representation where the limiting scale $J \rightarrow \infty$. In practice, we take a finite $J$ and therefore only have stability with respect to translation. The mathematical analysis for the stability properties of scattering networks is not limited to wavelets: for instance, it is generalized by  using semi-discrete frames as filters in \cite{WB15,WB16}, and time-frequency atoms as filters in \cite{CzaL17}. In all these cases, the scattering transforms are Lipschitz continuous with Lipschitz constant $L=1$, which is an important factor for the provable stability properties.

A scattering network extracts features from every convolutional layer. This is not the case for a general CNN. In \cite{Goodfellow16} a CNN is defined as a neural network which has at least one convolution unit. Many widely-adapted CNN models have either a sequential structure (e.g. the AlexNet \cite{Krizhevsky12}) or a more complex feed-forward structure (e.g. the GoogleNet \cite{Szegedy15}). For those models, stability is still an important issue. Intuitively, keeping the same energy in the feature, we should train the network so that the features are as stable as possible to small perturbations before using dense layers to do the classification. In \cite{SZSBEGF13}, the authors use the large Lipschitz bound of each single layer to illustrate that the AlexNet could be very unstable with respect to small perturbation on the input image. In fact, changing a small number of pixels could ``fool'' the network so that it produces wrong classification results. In general, a small Lipschitz bound of the entire transform implies the robustness of a CNN to small perturbations.

``Fooling'' networks is naturally connected to adversarial networks. Indeed, Lipschitz bounds are already used in training adversarial networks other than just quantitatively showing the robustness. In \cite{ArjCB17}, the authors propose an objective function for training generative adversarial networks where they use (the distance between) the Lipschitz constant (and $1$) as a penalty term. However, there is no direct way to impose it. Later in \cite{GulAA17}, the authors use a gradient penalty inspired by the fact that a function is $1$-Lipschitz if its gradient is bounded by $1$.

Although it plays an important role in deep learning, the study of Lipschitz bounds is not completely addressed by existing literature. The frameworks in \cite{Mallat12}--\cite{CzaL17} analyze the $1$-Lipschitz transformations but are limited to the scattering transforms and do not generalize automatically to general CNN's. \cite{SZSBEGF13} provides a Lipschitz bound using the product of Bessel bounds of each layer, but in general lacks tightness for non-sequential models such as the scattering network. Our paper fills in the gap between these approaches, by providing a unified stability analysis that applies to both the scattering networks (as in \cite{Mallat12}--\cite{CzaL17}) and to the more general convolutional networks. Our framework is flexible and compatible with architectures that may or may not generate outputs from hidden layers. The results presented in this paper are optimal for scattering networks and in general tighter than taking the product of Bessel bounds in each layer. Our focus is on estimation of these Lipschitz bounds, and how they relate to stochastic processes. We discuss how the Lipschitz bounds can be used for classification, but we do not focus on extending these results to generative adversarial networks. Instead we study numerically a few examples, including the AlexNet and the GoogleNet. Surprisingly, when applied to the AlexNet (and GoogleNet), we discovered that while the estimated bounds are about three orders of magnitude more conservative than the numerically estimated Lipschitz bounds, the empirical bounds are still three orders of magnitude smaller. Specifically, the largest local Lipschitz bound is obtained numerically to be of order 1, whereas on an extensive study using ImageNet \cite{ILSVRC15} images, the ratio between the energy of output variation to the energy of input variation is of the order $10^{-3}$. 

We first overview the CNN architecture considered in this paper (the details are given in the main test). The framework is applicable to the scattering network \cite{Mallat12,BM13}, the AlexNet \cite{Krizhevsky12} and the GoogleNet \cite{Szegedy15}. It can also be used to analyze models such as Long-Short Term Memory \cite{HS97}. We state the theory for continuous signals, but explain how to adapt it for the discrete case (which is the case for AlexNet and GoogleNet). We focus on the feature extraction part of the network and do not discuss the fully connected layers that are usually put on top of the structure, though the fully connected layers can be regarded as a special case of convolutional layers. The CNN that we consider has a feed-forward structure and consists of different layers (it is possible to use infinitely many layers to represent a feedback structure). We define the layers according to the convolutions. Specifically, each layer consists of input nodes, convolutional filters, detection / merge operations, pooling filters, output (feature) nodes and (hidden) output nodes.

\begin{itemize}
\item
The \emph{input nodes} are signals passed to the current layer. That could come from the hidden output nodes in the previous layer, or the input signal to the network.
\item
The \emph{convolutional filters} are the filters that perform convolution with the signal from the input nodes. Suppose $y$ is the signal in an input node, and $g$ is the convolutional filter, the output is
\[ z(t) = y \ast g(t) = \int y(t-s)g(s)ds = \int y(s)g(t-s)ds~. \]
\item
The \emph{pooling filters} are low-pass convolutional filters that lower the complexity before the feature is extracted as output. Note that these are still linear translation-invariant operations which are commonly used in scattering networks. The nonlinear operations such as max pooling and average pooling are contained in the detection operations. 
\item
The \emph{(feature) output nodes} are outputs of the convolutional neural network. As we specified earlier, these nodes form a subset of the representation. Once the representation is extracted, the specific machine learning tasks, such as classification and prediction, will be performed on  the representation. 
\item
The \emph{dilation} operations are ``changes of scale'' on the space variables. A dilation operation on a signal $f(x)$, $x \in \RR^d$, can be represented using a $d \times d$ invertible matrix $D$. The dilated signal is $f(Dx)$.
\item
The \emph{detection} operations are nonlinear operations that apply pointwise to the output of the convolutional filters. The nonlinearities have Lipschitz constant $1$ (e.g. ReLU functions). In addition to applying the nonlinearity, the outputs can be aggregated by \emph{merge} operations to produce a single output for dimensionality reduction. The max pooling and average pooling are modeled in this manner. 
\item
The \emph{(hidden) output nodes} are signals that propagate to the next layer. The signals at the output nodes are identical to those at the input nodes of the next layer.
\end{itemize}

\begin{figure}[!ht]
\centering
\includegraphics[width=0.75\linewidth]{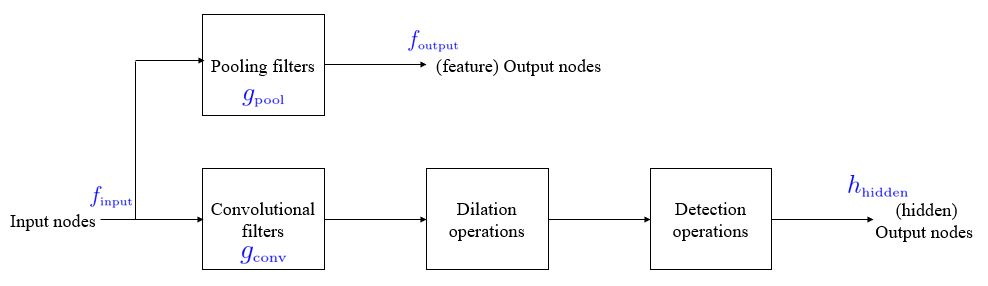}
\caption{The structure of a network layer. The network we consider consists of a number of layers, which makes the structure ``deep''.}
\label{fig:layer_structure}
\end{figure}

In this paper, unless otherwise specified, we use $f$ to denote the input and output signals of a CNN, $h$ to denote the hidden features, and $g$ to denote filters. The input signal on the $d$-dimensional Euclidean space has finite energy, that is, $f \in L^2(\RR^d)$. The Fourier transform of $f$, denoted by $\hat{f}$, is defined formally to be \[\hat{f}(\omega) = \int_{\RR^d} f(x) e^{-2 \pi i \omega \cdot x} dx ~, \quad \omega \in \RR^d ~.\] and we refer the readers to \cite{Ben96} for rigorous definitions for $f$ when $f \in L^2(\RR^d)$ or when $f$ is a generalized function.  The filters of CNN are taken from the Banach Algebra of tempered distributions with an essentially bounded Fourier Transform, that is,
\begin{equation}
\label{def:balg}
\cB = \left\{ g \in \cS'(\RR ^d), \norm{\hat{g}}_{\infty} < \infty \right\} ~.
\end{equation}
We have a detailed discussion of this algebra in Appendix \ref{appendix:banach}. We use $\norm{\cdot}_p$ to denote the $L^p$-norm corresponding to the Lebesgue integral. For a matrix $A$, $A^t$ denotes its transpose, and $A^*$ denotes its conjugate transpose. We use $\norm{A}_{\op} = \max_{\norm{x}_2=1} \norm{Ax}_2$ to denote the operator norm of $A$, $\norm{A}_{\ast} = \tr(\sqrt{A^* A})$ to denote its nuclear norm, and $\norm{A}_{\Fr} = \sqrt{\tr(A^* A)}$ to denote its Frobenius norm.

The paper is organized as follows. Section \ref{sec:definecnn} sets up the mathematical problem  by defining the layers of a CNN. Section \ref{sec:lipschitz} states the results on estimating the Lipschitz bounds. Section \ref{sec:examples}  illustrates examples from the scattering network to the AlexNet and the GoogleNet. Section \ref{sec:stationary} discusses how the Lipschitz bounds relate to concentration results for stationary processes on CNN's. 
Section \ref{sec:classification} discusses using the Lipschitz bounds to construct a nonlinear discriminant.

\section{Defining a CNN}\label{sec:definecnn}

The overall structure of an $M$-layer CNN is illustrated in Figure \ref{fig:layerdetail}. The picture shows how an input propagates through the layers while generating outputs at each layer. The details of the layers are described in the following two subsections. If no merge operation is present at a certain layer, the convolutional layer is modeled as a linear operation followed by nonlinearity; if there are merge operations, different types of merge operations are modeled separately. 

\begin{figure*}[!ht]
\centering
\includegraphics[width=0.75\linewidth]{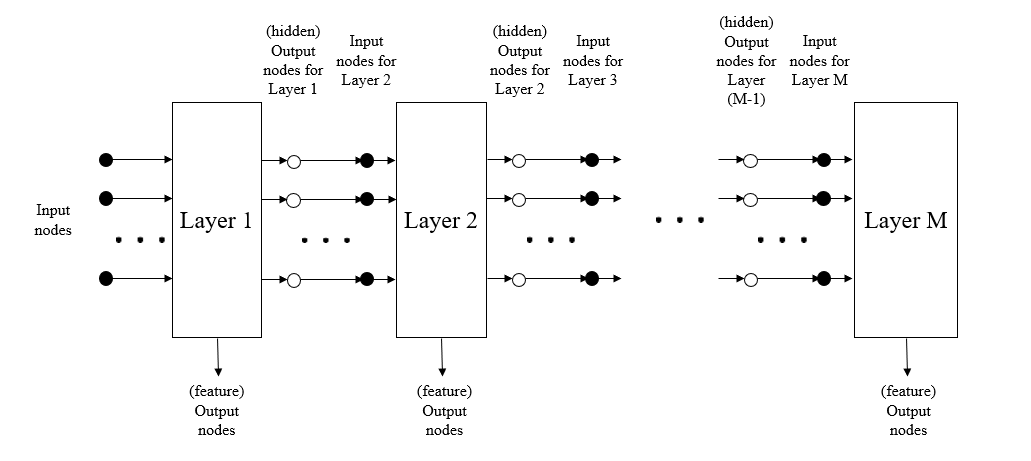}
\caption{The detail of an $M$-layer CNN. The signals at output nodes are identical as at input nodes in the next layer. There may or may not be output generation in each layer.}
\label{fig:layerdetail}
\end{figure*}

\subsection{A layer without merge operations}
If a certain layer does not contain any merging, we can model the filters as a linear transform from signals on all the input nodes. In the $m$-th layer, the set of input nodes is denoted by $\cI_m = \{N_{m,1}$, $N_{m,2}$, $\cdots$, $N_{m,n_m}\}$ and the set of output nodes by $\cO_m = \{N'_{m,1}$, $N'_{m,2}$, $\cdots$, $N'_{m,n'_m}\}$. Further, the set of output generating nodes is denoted by $\cV_m = \{V_{m,1}$, $V_{m,2}$, $\cdots$, $V_{m,n_m}\}$. With this notation, let $h_{m,1}$, $h_{m,2}$, $\cdots$, $h_{m,n_m}$ be the signals on the input nodes, a linear operator $T^{(m)}$ is a $n'_m$-by-$n_m$ array of filters $T^{(m)}_{n',n}$ in $\cB$ such that \[h^{\spadesuit}_{m,n'} = \sum_{n=1}^{n_m} T^{(m)}_{n',n} \ast h_{m,n}~, \quad 1 \leq n' \leq n'_m,\] is received before downsampled by the $d$-by-$d$ invertible matrix $D_{m,n'}$ and sent into a nonlinearity $\sigma_{m,n'}$ to output \[h'_{m,n'} (x) = \sigma_{m,n'} \left(  h^{\spadesuit}_{m,n'} (D_{m,n'} x) \right) ~.\]

\begin{figure}[!ht]
\centering
\includegraphics[width=0.75\linewidth]{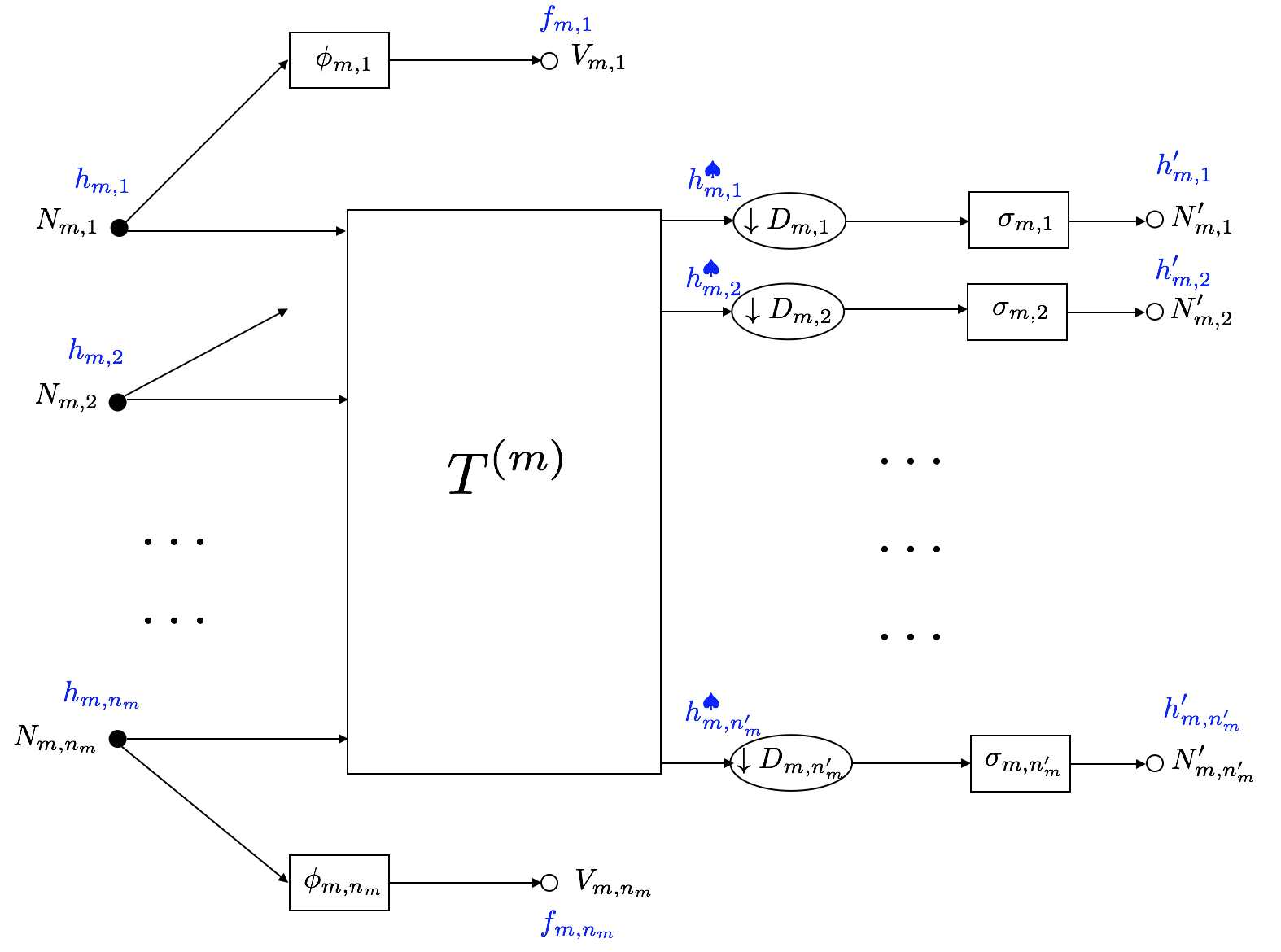}
\caption{The detail of the m-th layer with no merge operations. $N_{m,n}$ denote the input nodes, $N'_{m,n'}$ denote the hidden output nodes, $V_{m,n}$ denote the feature output nodes. $\phi_{m,n}$ denote the pooling filters, $D_{m,n'}$ denote the dilation factors, and $\sigma_{m,n'}$ denote the 1-Lipschitz nonlinearities. The notations in blue represent the signals at each node. $h_{m,n}$ denote the input signals of the layer. $h'_{m,n'}$ denote the hidden output signals that are passed to the next layers. $h^{\spadesuit}_{m,n'}$ denote the signals received after passing the linear operator $T^{(m)}$. $f_{m,n}$ denote the signals at the feature output nodes.}
\label{fig:nomerge}
\end{figure}

For the $m$-th layer, we define three types of Bessel bounds as follows.
For each $\omega \in \RR^d$, denote $\hat{T}^{(m)}(\omega)$ to be the $n'_m \times n_m$ matrix that contains the Fourier transform $\hat{T}^{(m)}_{n',n}$ of $T^{(m)}_{n',n}$ at $\omega$, for $1 \leq n \leq n_m$, $1 \leq n' \leq n'_m$. Also for each $\omega$, denote $\hat{\Psi}^{(m)}(\omega)$ to be the $n_m \times n_m$ diagonal matrix that has $\hat{\phi}_{m,n}(\omega)$, the Fourier transform of the convolutional filter at $\omega$, as its $(n,n)$ entry. Let $\D^{(m)}$ be the $n'_m \times n'_m$ diagonal matrix with $(\det D_{m,n'})^{-1/2}$ as its $(n',n')$ entry. The 1st type Bessel bound for the $m$-th layer is defined to be
\begin{equation}\label{def:b1mnm}
B_m^{(1)} = \sup_{\omega \in \RR^d} \norm{\begin{bmatrix}
 \D^{(m)} \hat{T}^{(m)}(\omega)\\
 \hat{\Psi}^{(m)}(\omega)
\end{bmatrix}}_{\op}^2 ~,
\end{equation}
the 2nd type Bessel bound for the $m$-th layer is defined to be
\begin{equation}\label{def:b2mnm}
B_m^{(2)} = \sup_{\omega \in \RR^d} \norm{
 \D^{(m)} \hat{T}^{(m)} (\omega)
 }_{\op}^2 ~,
\end{equation}
and the 3rd type Bessel bound is defined to be 
\begin{equation}\label{def:b3mnm}
B_m^{(3)} = \sup_{\omega \in \RR^d} \norm{\hat{\Psi}^{(m)}(\omega)}_{\op}^2 ~.
\end{equation}

In general, the Bessel bound quantifies how the energy is magnified by  convolution. The bound is finite if the filters form semi-discrete frames (see \cite[Appendix A]{WB16}). Our definition acts in the spectral domain and it naturally yields estimates of the the Lipschitz bounds: see Appendix A,  (\ref{eq:calBessel}). The need for three types of Bessel bounds is related to different types of energy mixing: input-to-combined hidden and feature output nodes, input-to-hidden output nodes, and input-to-feature output nodes.
Intuitively, $B_M^{(1)}$ is the Bessel bound for the frame composed of both $T_{n', n}^{(m)}$ and $\phi_{m,n}$, $B_M^{(2)}$ is for the frame of $T_{n', n}^{(m)}$ and $B_M^{(3)}$ is for the frame of $\phi_{m,n}$ only. For a layer with merge operations, the Bessel bounds share the same intuition, but their estimates have different mathematical representations. We describe that in the next section. 

\subsection{A layer with merge operations}
There are three types of merging. Type I takes inputs $y_1, \cdots, y_k$ from $k$ channels, applies a nonlinearity function $\sigma_1, \cdots, \sigma_k$ respectively, and then sums them up. That is, the output is
\begin{equation}
z = \sum_{j=1}^k \sigma_j (y_j) ~.
\end{equation}
Type II takes inputs $y_1, \cdots, y_k$ from $k$ channels, apply a nonlinearity on each signal, and then aggregates them by a pointwise $p$-norm. That is, the output is
\begin{equation}
z = \left( \sum_{j=1}^k \abs{ \sigma_j (y_j) }^p \right) ^ {1/p} , ~ \textup{if } p < \infty ~;
\end{equation}
and
\begin{equation}
z = \max_{j=1,\cdots,k} \abs{ \sigma_j (y_j) }, ~ \textup{if } p = \infty ~.
\end{equation}
Type III takes inputs $y_1, \cdots, y_k$ from $k$ channels, apply a nonlinearity on each signal, and then performs a pointwise multiplication. The nonlinearity $\sigma_j$ should satisfy $\norm{\sigma_j}_{\infty} \leq 1$ for each $j$. The output is
\begin{equation}
z = \prod_{j=1}^k \sigma_j(y_j) ~.
\end{equation}

\begin{figure}[!ht]
\centering
\includegraphics[width=0.75\linewidth]{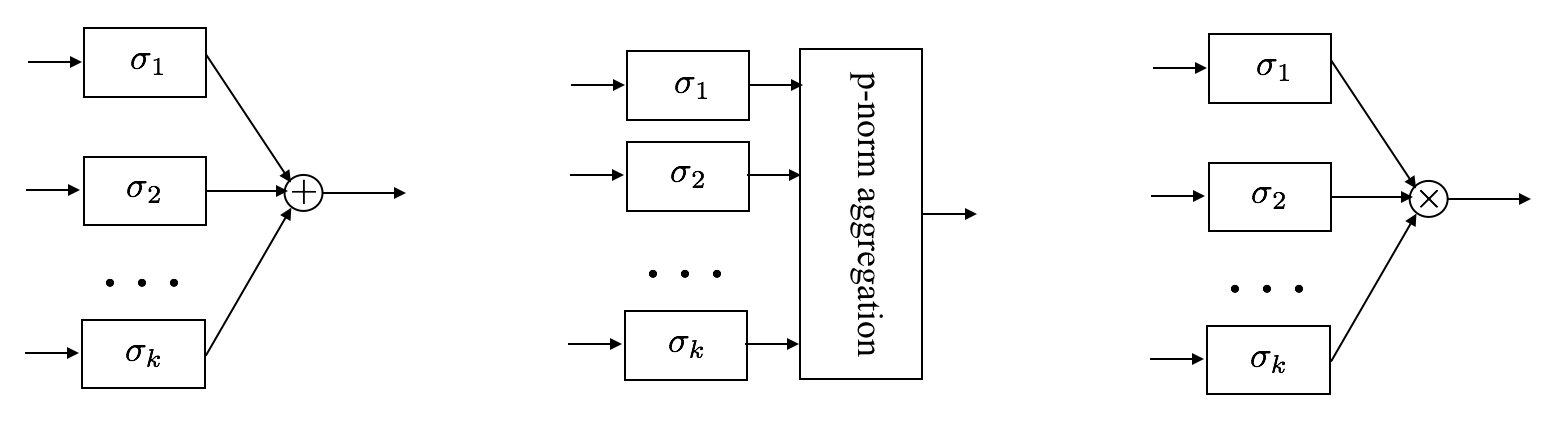}
\caption{The three types of merge. Left: Type I - taking sum of the inputs; middle: Type II - taking $p$-norm aggregation of the inputs; right: Type III - taking pointwise product of the inputs.}
\label{fig:layer}
\end{figure}

We point out that the standard pooling operations in most discrete CNN's can be modeled in the continuous case by these merge operations. Specifically, \emph{max pooling} is the operation of taking the maximal element among those in the same sub-regions. We can use translations and dilations to separate elements in a sub-region to distinct channels, as illustrated in Figure \ref{fig:maxpool}. Then the $L^{\infty}$-aggregation select the largest element and performs the max pooling. \emph{Average pooling} replaces ``taking the max'' by ``taking the average''. Similarly to max pooling, it can be done by taking the sum as illustrated in Figure \ref{fig:avgpool}. A concrete example illustrates max pooling as implemented by this framework. Similar implementation can realize average pooling. Consider the finite signal $(1,3,4,2,1,5,6,7)$ in Figure \ref{fig:exMaxPooling} for which we want to apply max pooling with size = 2 and stride = 2. Then the max pooled signal is $(3,4,5,7)$, where each entry is the larger value within each pair. Consider now the (circular) translation by 1 pixel of the first signal, that is $(3,4,2,1,5,6,7,1)$ together with the original signal (the middle two signals in the figure). Apply the dilation operator where we discard the second pixel in each consecutive pair of pixels. Thus we obtain $(1,4,1,6)$ and $(3,2,5,7)$ respectively. Now a Type II aggregation with $p = \infty$ selects the larger value between two pixels at the same position, and therefore results in $(3,4,5,7)$, which is the same as the max pooling operation applied on the original signal.


\begin{figure}[ht!]
\begin{subfigure}{.5\textwidth}
  \centering
  \includegraphics[width=\linewidth]{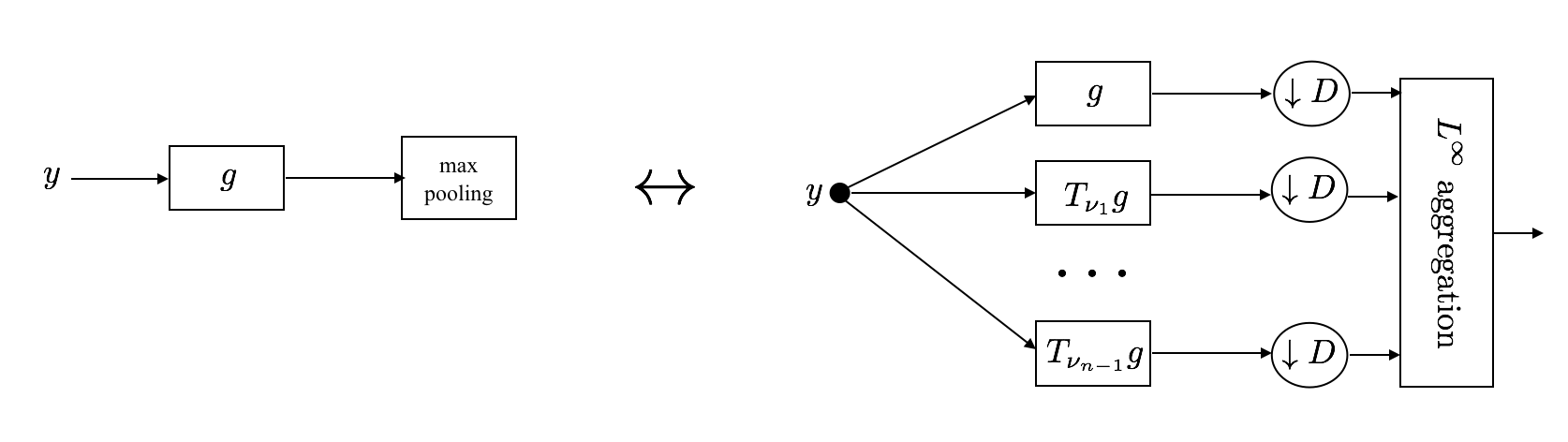}
\caption{max pooling}
\label{fig:maxpool}
\end{subfigure}
\begin{subfigure}{.5\textwidth}
  \centering
  \includegraphics[width=\linewidth]{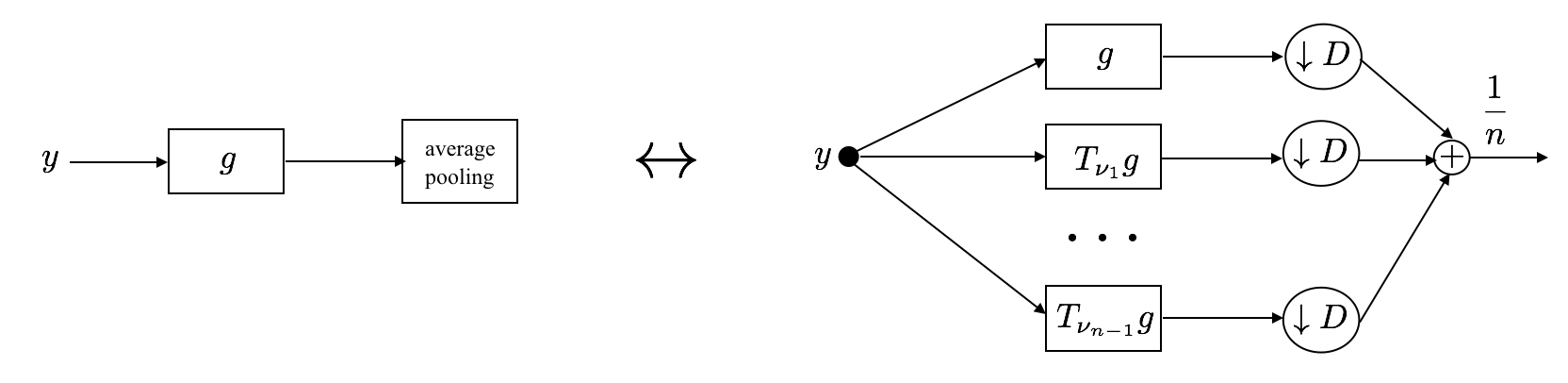}
\caption{ average pooling }
\label{fig:avgpool}
\end{subfigure}
\caption{In the continuous case, the max pooling is modeled as Type II aggregation for $p = \infty$, and the average pooling is modeled as Type I aggregation.}
\end{figure}

\begin{figure}[!ht]
\centering
\includegraphics[width=0.75\linewidth]{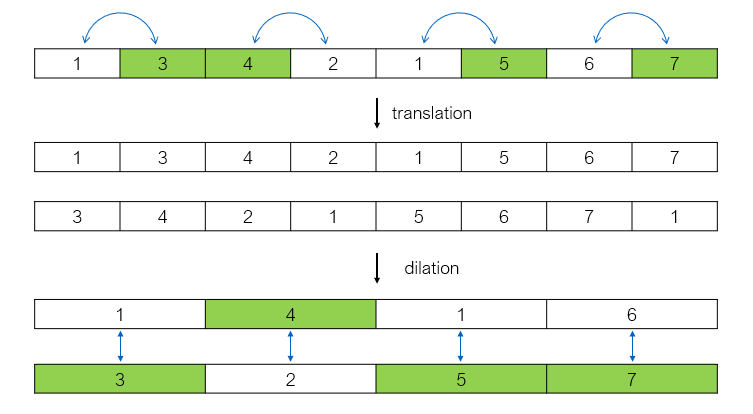}
\caption{A concrete example for the max pooling.}
\label{fig:exMaxPooling}
\end{figure}

Suppose there are $n_m$ nodes in the $m$-th layer (this works for $m < M$ but $m=M$ is a similar case in which there is no hidden output node). The set of these input nodes is denoted by $\cI_m = \{ N_{m,1}, N_{m,2}, \cdots, N_{m,n_m} \}$. Within the layer, each node is connected to several filters. The filter can be either a pooling filter, or a convolutional filter. Associated with $N_{m,n}$ for $1 \leq k \leq n_m$, the pooling filter is denoted to be $\phi_{m,n}$, and the convolutional filters to be $G_{m,n} = \{ g_{m,n;1}, \cdots g_{m,n;k_{m,n}} \}$. The set of filters in the $m$-th layer is thus
\begin{equation}
G_m = \cup_{n=1}^{n_m} G_{m,n} ~.
\end{equation}
Each filter $g_{m,n;k_{m,n}}$ is naturally classified into one of three categories according to the three types of merging: if a filter is merged using Type I operation, then it is classified as a Type I filter; in the same manner we define Type II and Type III filters. If a filter is not merged with other filters, we classify it as Type I (with $k=1$ in the first picture in Figure (\ref{fig:layer})). We denote the sets of all Type-I, II, III filters by $\tau_1, \tau_2, \tau_3$, respectively.

Note that each filter is associated with one and only one output node. Let $\cO_m = \{ N'_{m,1}, N'_{m,2}, \cdots, N'_{m,n'_m} \}$ denote the set of output nodes of the $m$-th layer. Note that $n'_m = n_{m+1}$ and there is a one-one correspondence between $\cO_m$ and $\cI_{m+1}$. The output nodes automatically divide $G_m$ into $n'_m$ disjoint subsets $G_m = \cup_{n'=1}^{n'_m} G'_{m,n'}$, where $G'_{m,n'}$ is the set of filters merged into $N'_{m,n'}$. Further, $\cV_m = \{V_{m,1}, V_{m,2}, \cdots, V_{m,n_m}\}$ denote the set of output generating nodes. The detail of one layer is illustrated in Figure \ref{fig:onelayerdetail}.

\begin{figure*}[!ht]
\centering
\includegraphics[width=0.75\linewidth]{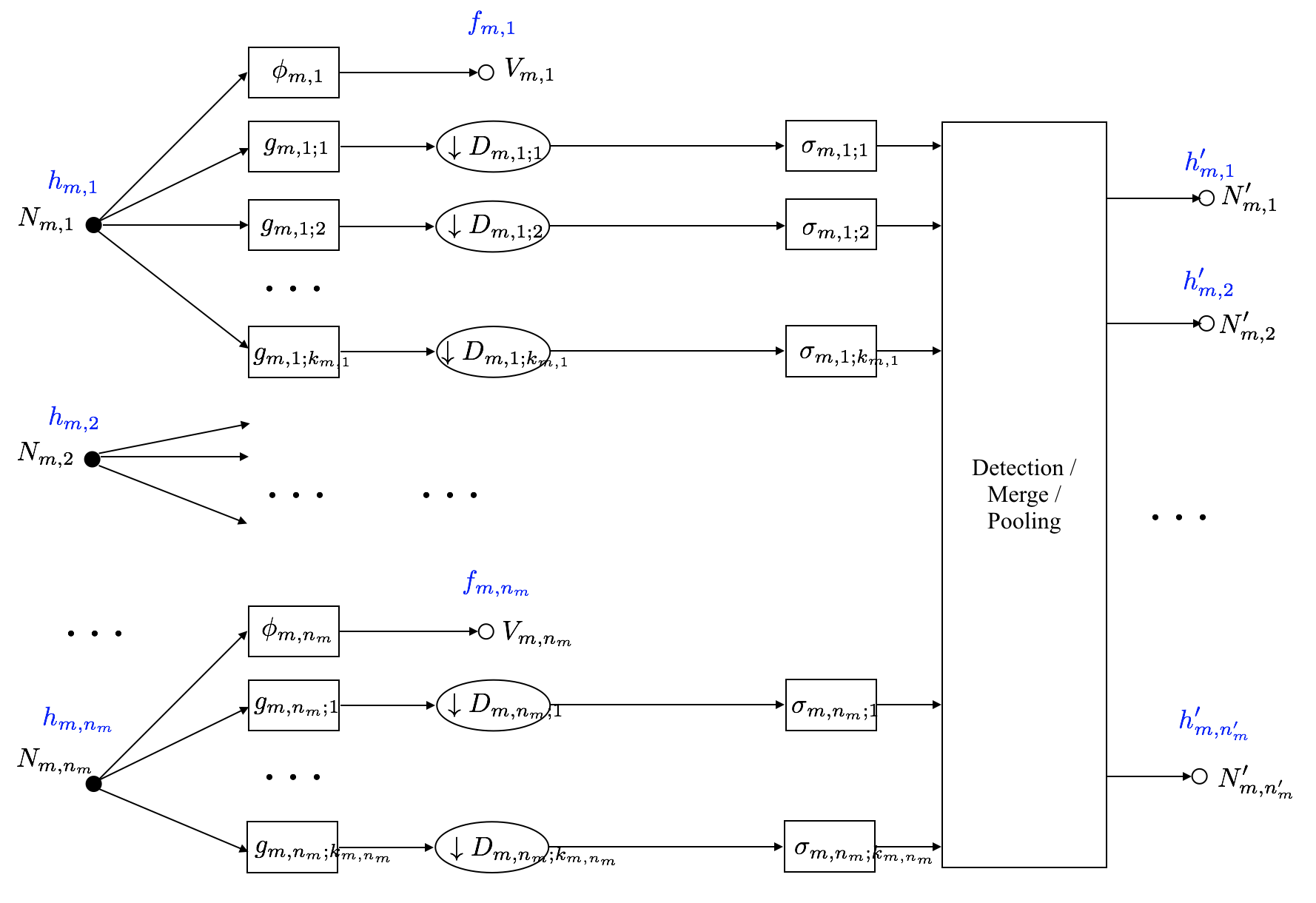}
\caption{The detail of one layer with merging. $N_{m,n}$ denote the input nodes, $N'_{m,n'}$ denote the output nodes, $V_{m,n}$ denote the output generating nodes. $\phi_{m,n}$ and $g_{m,n}$ denote the filters, $D_{m,n;k}$ denote the dilation factors. $\sigma_{m,n;k}$ denote the 1-Lipschitz nonlinearities (for illustration we put them outside the merge box, but they belong to the merge operations where we defined the three types of merge). The notations in blue represent the signals at the nodes. $h_{m,n}$ denote the input signals of the layer. $h'_{m,n'}$ denote the output signals that are passed to the next layers. $f_{m,n}$ denote the signals at the feature output nodes.}
\label{fig:onelayerdetail}
\end{figure*}

For each filter $g_{m,n;k}$, we define the associated multiplier $l_{m,n;k}$ in the following way: suppose $g_{m,n;k} \in G'_{m,n'}$, let $K = \abs{G'_{m,n'}}$ denote the cardinality of $G'_{m,n'}$. Then
\begin{equation}
l_{m,n;k} =
\begin{cases}
K & \text{, if } g_{m,n;k} \in \tau_1 \cup \tau_3\\
K^{\max \{ 0, 2/p-1 \}} & \text{, if } g_{m,n;k} \in \tau_2
\end{cases}
\end{equation}

We define the 1st type Bessel bound for the node $N_{m,n}$ to be
\begin{equation}
\label{def:b1mn}
B^{(1)}_{m,n} = \norm{ \abs{\hat{\phi}_{m,n}}^2 + \sum_{k=1}^{k_{m,n}} l_{m,n;k} D_{m,n;k}^{-d} \abs{\hat{g}_{m,n;k}}^2 }_{\infty} ~,
\end{equation}
the 2nd type Bessel bound to be
\begin{equation}
\label{def:b2mn}
B^{(2)}_{m,n} = \norm{ \sum_{k=1}^{k_{m,n}} l_{m,n;k} D_{m,n;k}^{-d} \abs{\hat{g}_{m,n;k}}^2 }_{\infty} ~,
\end{equation}
and the 3rd type Bessel bound to be
\begin{equation}
\label{def:b3mn}
B^{(3)}_{m,n} = \norm{ \hat{\phi}_{m,n} }_{\infty}^2 ~.
\end{equation}
Further, we define the 1st type Bessel bound for the $m$-th layer to be
\begin{equation}
\label{def:b1m}
B^{(1)}_{m} = \max_{1 \leq n \leq n_m} B^{(1)}_{m,n} ~,
\end{equation}
the 2nd type Bessel bound to be
\begin{equation}
\label{def:b2m}
B^{(2)}_{m} = \max_{1 \leq n \leq n_m} B^{(2)}_{m,n} ~,
\end{equation}
and the 3rd type Bessel bound to be
\begin{equation}
\label{def:b3m}
B^{(3)}_{m} = \max_{1 \leq n \leq n_m} B^{(3)}_{m,n} ~.
\end{equation}

\section{Calculating the Lipschitz bound}\label{sec:lipschitz}
Suppose we are given with a CNN within the framework given in Section \ref{sec:definecnn}. For any input signal $f$ and $\tilde{f}$, let $f_N$ be the output for $f$ from the node $N$, and $\tilde{f}_N$ be the output for $\tilde{f}$ from the node $N$. Let $\cV = \cup_{m=1}^M \cV_m$ be the collection of all output generating nodes. We say $L$ is a \emph{Lipschitz bound} for the CNN if
\begin{equation}
\sum_{N \in \cV} \norm{f_N - \tilde{f}_N}_2^2 \leq L \norm{f - \tilde{f}}_2^2 ~.
\end{equation}

The map $\Phi: L^2(\RR^d) \rightarrow [L^2(\RR^d)]^{\abs{\cV}}$ induced by the CNN is defined by\begin{equation}\label{def:Phi}\Phi(f) = (f_N)_{N \in \cV} ~.\end{equation}A norm $|||\cdot|||$ defined on $[L^2(\RR^d)]^{\abs{\cV}}$ by \[\Big|\Big|\Big|(f_N)_{N \in V}\Big|\Big|\Big| = \left( \sum_{N \in \cV} \norm{f_N}_2^2 \right)^{1/2}\]is well defined and $L_c = \sqrt{L}$ is a \emph{Lipschitz constant} in the sense that\begin{equation}\Big|\Big|\Big| \Phi(f) - \Phi(\tilde{f}) \Big|\Big|\Big| \leq L_c \norm{f-\tilde{f}}_2 ~.\end{equation}
We have the following theorem for calculating the Lipschitz bound.

\begin{theorem}
\label{thm:lp}
Consider a CNN in the framework of Section II, with $M$ layers and in the $m$-th layer it has 1st type Bessel bound $B_m^{(1)}$, 2nd type Bessel bound $B_m^{(2)}$ and 3rd type Bessel bound $B_m^{(3)}$. Then the CNN induces a nonlinear map $\Phi$ that is Lipschitz continuous, and its Lipschitz bound is given by the optimal value of the following linear program:
\begin{equation}
\label{eq:lp}
\begin{aligned}
\max \quad & \sum_{m=1}^M z_m \\
\textup{s.t.} \quad & y_0 = 1 \\
& y_m + z_m \leq B^{(1)}_m y_{m-1}, \quad 1 \leq m \leq M-1 \\ 
& y_m \leq B^{(2)}_m y_{m-1}, \quad 1 \leq m \leq M-1 \\
& z_m \leq B^{(3)}_m y_{m-1}, \quad 1 \leq m \leq M \\
& y_m, z_m \geq 0, \quad \mbox{for all} ~ m ~.
\end{aligned}
\end{equation}
\end{theorem}

The proof of Theorem \ref{thm:lp} is given in Appendix \ref{appendix:prooflip}. We  remark here that the linear program presented as (\ref{eq:lp}) is feasible, since one obvious feasible point is $y_m = 0$ for $1 \leq m \leq M-1$ and $z_m = 0$ for $1 \leq m \leq M$. Moreover, the solution is bounded since all $z_m$'s are bounded by $B_m^{(3)} \prod_{m'=1}^{m-1} B_{m'}^{(2)}$ according to the third and fourth inequalities in (\ref{eq:lp}). In practice, either the simplex method or the interior method (see, for instance \cite[Chapter 13-14]{NocW99}) can be used to solve this linear program, and they run in polynomial time with respect to the number of layers. If we are in the discrete case, say for pixel images, then we need to compute the Bessel bounds, which relies on the Fast Fourier Transforms that grows as $O(N \log N)$ with the dimensionality of filters. Although the complexity is not high, a Lipschitz bound computed via a linear program is still not intuitive. We give more explicit estimates of the Lipschitz bound in the following corollaries.
\begin{corollary}
\label{thm:prod}
Consider a CNN in the framework of Section II, with $M$ layers and in the $m$-th layer it has 1st type Bessel bound $B_m^{(1)}$. Then the CNN induces a nonlinear map that is Lipschitz continuous, and its Lipschitz bound is given by
\begin{equation}
\prod_{m=1}^M \max \{1,B_m^{(1)}\} ~.
\end{equation}
\end{corollary}

\begin{corollary}
\label{thm:sumprod}
Consider a CNN in the framework of Section II, with $M$ layers and in the $m$-th layer it has 2nd type Bessel bound $B_m^{(2)}$ and generating bound $B_m^{(3)}$. Then the CNN induces a nonlinear map that is Lipschitz continuous, and its Lipschitz bound is given by
\begin{equation}
B_1^{(3)}+\sum_{m=2}^{M} B_m^{(3)} \prod_{m'=1}^{m-1} B_m^{(2)} ~.
\end{equation}
\end{corollary}

The proof of Corollary \ref{thm:sumprod} is an immediate consequence of Theorem \ref{thm:lp}, specifically from the third and fourth inequalities of (\ref{eq:lp}). The proof of Corollary \ref{thm:prod} is given in Appendix \ref{appendix:prooflipcor}. We remark here that both corollaries give a more conservative bound compared to the linear program (\ref{eq:lp}) because both results restrict the variables to a subset of the feasible region. The idea of using Bessel bounds is also addressed in \cite{SZSBEGF13} where the authors compute the Bessel bounds of each layer of the AlexNet, and in \cite{WB16} where the authors set $B_m \leq 1$ to make the CNN a 1-Lipschitz map. We return to the AlexNet in the following section.

Subject to the knowledge of the three types of Bessel bounds in each layer, the estimate given by the linear program (\ref{eq:lp}) is tight. However three issues may prevent its tightness.
 First, except for the scattering network when defined for continuous inputs, most of CNN's consider discrete time inputs only. Second, even subject to the same Bessel bounds, different filters may produce much smaller Lipschitz bounds. Sub-optimality occurs in cases where the signal that achieves the Bessel bound for Layer $m+1$ is not in the range of Layer $m$. Third, in some practical applications when signals are modeled as samples drawn from certain distributions, then the emphasis is on local stability around the operating distributions, whereas the global Lipschitz bound may be irrelevant. 
 
 We address these issues by looking at three examples: the scattering network, a toy network that includes all three types of merge operations we consider in this paper, and the well-known AlexNet and GoogleNet.





\section{Examples}\label{sec:examples}

\subsection{Scattering network}\label{subsec:scattering}

The scattering network in \cite{Mallat12, BM13}
is a 1-Lipschitz map. In each layer the filters are designed to form wavelet orthonormal bases using multi-resolution analysis. Such design leads to $B_{m,n}^{(1)} = B_{m,n}^{(2)} = B_{m,n}^{(3)} = 1$, for all $m,n$. Then Corollary \ref{thm:prod} simply yields a Lipschitz bound $L = 1$ which is tight. We refer the readers to \cite[Section 4.1]{BSZ17}  for a detailed discussion.

\subsection{A toy example that contains merge operations}\label{subsec:merge}


The scattering network enjoys $B_{m,n}^{(1)} = B_{m,n}^{(2)} = B_{m,n}^{(3)} = 1$ for all $m,n$ since it is tightly related to wavelet decompositions. In many CNN's we don't have feature output from hidden layers and therefore $B_{m,n}^{(1)} = B_{m,n}^{(2)}$, whence the results in Corollary \ref{thm:prod} coincide with the optimal value by the linear program (\ref{eq:lp}). However, Corollary \ref{thm:prod} can be suboptimal. To see this, we take a toy example of CNN that contains merge operations. The same network structure appears also in \cite{BSZ17} with different filter weights. The parameter $p$ is set to $p=2$.

Figure \ref{fig:example} is an illustration of the CNN. According to Appendix \ref{appendix:banach}, we can translate it into a CNN within our framework, as illustrated in Figure \ref{fig:exampleequiv}.

\begin{figure}[!ht]
\centering
'\includegraphics[width=0.75\linewidth]{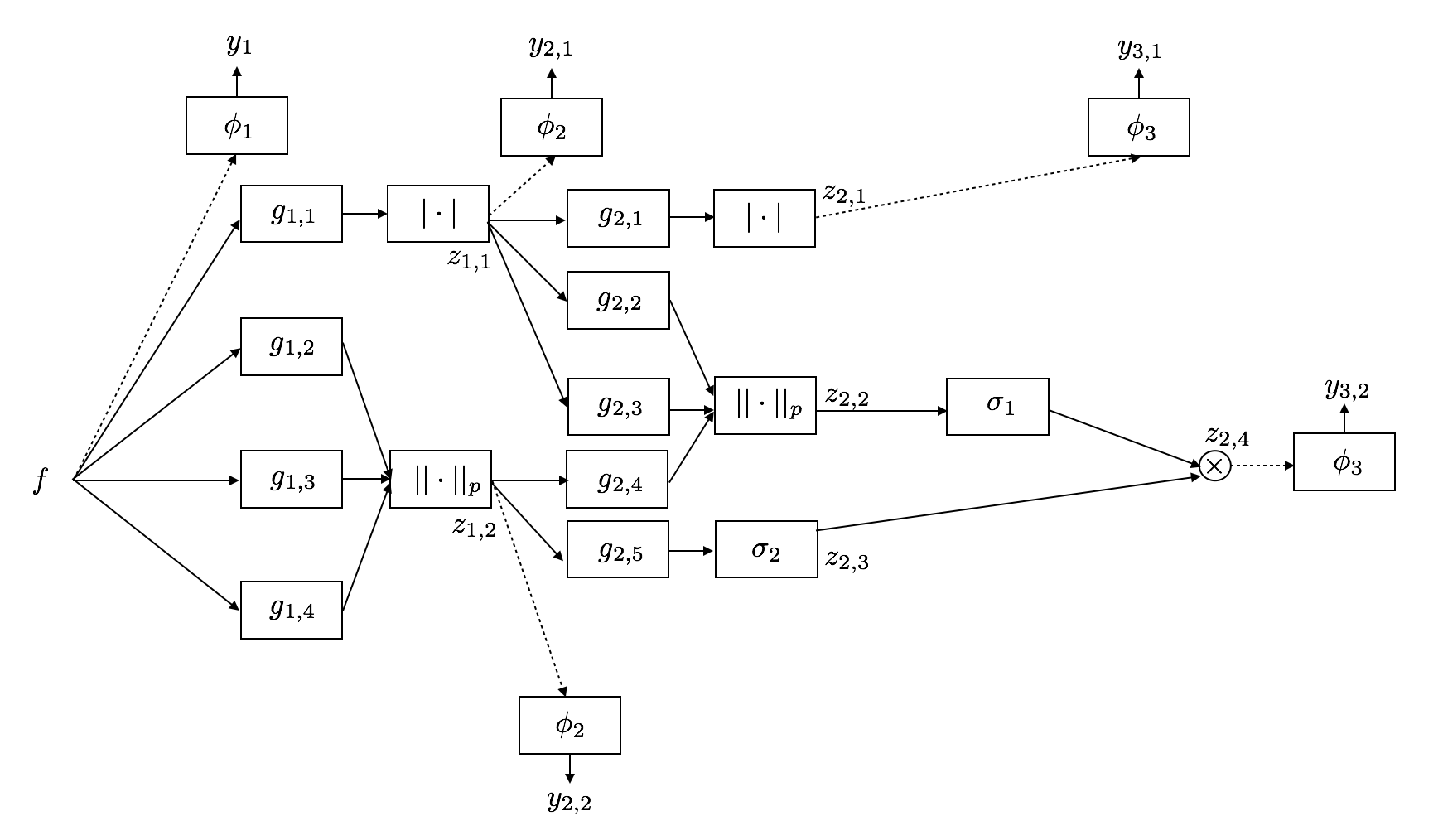}
\caption{The toy example that also appears in \cite{BSZ17}. Note that we have different choices of filters in the numerical experiment.}
\label{fig:example}
\end{figure}

\begin{figure}[!ht]
\centering
\includegraphics[width=0.75\linewidth]{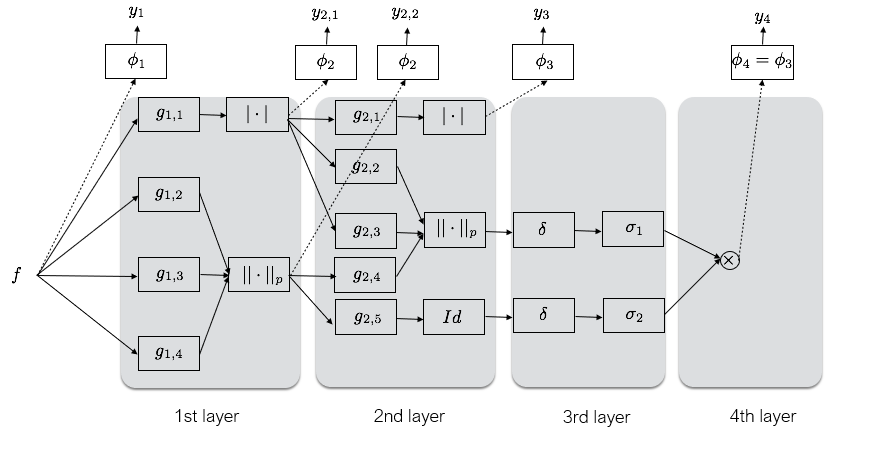}
\caption{Equivalent representation of the CNN. We illustrate the four layers of the network.}
\label{fig:exampleequiv}
\end{figure}

Define the smooth ``gate'' function on the Fourier domain supported on $(-1,1)$ as
\begin{equation}
\begin{aligned}
F(\omega) = & \exp\left( \frac{4\omega^2+4\omega+1}{4\omega^2+4\omega} \right) \chi_{(-1,-1/2)}(\omega) + 
\chi_{(-1/2,1/2)}(\omega) + 
\exp \left( \frac{4\omega^2-4\omega+1}{4\omega^2-4\omega} \right) \chi_{(1/2,1)}(\omega) ~.
\end{aligned}
\end{equation}
With this, we define the Fourier transforms of the filters to be $C^{\infty}$ gate functions
\begin{equation}
\begin{aligned}
\hat{\phi}_1(\omega) ~=~ & F(\omega) \\
\hat{g}_{1,j}(\omega) ~=~ & F(\omega+2j-1/2) + F(\omega-2j+1/2), 
\qquad j = 1,2,3,4. \\
\hat{\phi}_2(\omega) ~=~ & \exp\left(\frac{4\omega^2+12\omega+9}{4\omega^2+12\omega+8}\right) \chi_{(-2,-3/2)}(\omega) + 
\chi_{(-3/2,3/2)}(\omega) + 
\exp\left(\frac{4\omega^2-12\omega+9}{4\omega^2-12\omega+8}\right) \chi_{(3/2,2)}(\omega) \\
\hat{g}_{2,j}(\omega) ~=~ & F(\omega+2j) + F(\omega-2j) ,
\qquad j = 1,2,3. \\
\hat{g}_{2,4}(\omega) ~=~ & F(\omega+2)+F(\omega-2) \\
\hat{g}_{2,5}(\omega) ~=~ & F(\omega+5)+F(\omega-5) \\
\hat{\phi}_3(\omega) ~=~ & \exp \left( \frac{4\omega^2+20\omega+25}{4\omega^2+20\omega+24} \right) \chi_{(-3,-5/2)}(\omega) + 
\chi_{(-5/2,5/2)}(\omega) + 
\exp\left(\frac{4\omega^2-20\omega+25}{4\omega^2-20\omega+25}\right) \chi_{(5/2,3)}(\omega) ~.
\end{aligned}
\end{equation}

\begin{figure}[!ht]
\centering
\includegraphics[width=0.6\linewidth]{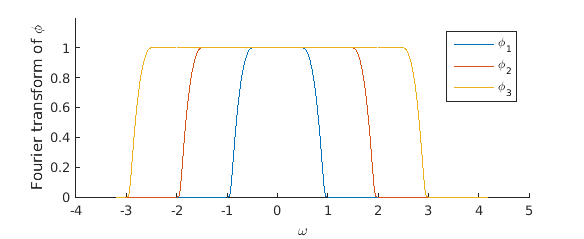}
\caption{Illustration of the filters $\phi_1$, $\phi_2$ and $\phi_3$ in the frequency domain. Note that they are all $C^{\infty}$ smooth functions.}
\label{fig:filters}
\end{figure}

Table I lists the Bessel bounds for all the layers. The optimal value of the linear program (\ref{eq:lp}) gives a Lipschitz bound of $L=2.866$; the Lipschitz bound as derived in Corollary \ref{thm:prod} is $L=8 [\exp(-1/3)]^2 = 4.102$; Corollary \ref{thm:sumprod} gives an estimate of the Lipschitz bound of $L=5$. We see that the output of the linear program (\ref{eq:lp}) is more optimal than the product given in Corollary \ref{thm:prod} and \ref{thm:sumprod}.

\linespread{1.5}

\begin{table}[ht!]
\centering
 \begin{tabular}{||c | p{0.8cm} p{0.8cm} p{0.8cm} c ||} 
 \hline
 $m$ & 1 & 2 & 3 & 4 \\ [0.5ex] 
 \hline\hline
 $B_m^{(1)}$ & $2e^{-1/3}$ & $2e^{-1/3}$ & 2 & 1 \\
 \hline
 $B_m^{(2)}$ & 1 & 1 & 2 & 0 \\
 \hline
 $B_m^{(3)}$ & 1 & 1 & 1 & 1 \\ [1ex] 
 \hline
 \end{tabular}
\label{tab:bessel} 
\vspace{1mm}
\caption{The Bessel bounds of the example in Figure \ref{fig:example}.}
\end{table}

\linespread{1}

\subsection{AlexNet and GoogleNet}\label{subsec:alexgoogle}

In this subsection we analyze the Lipschitz properties of AlexNet and GoogleNet. 
First, we apply the analytical results derived in earlier sections to these networks and compare their results to empirical estimates. To accomplish this, we need to extend the theory hitherto developed to processing of discrete signals. Second, we construct a local Lipschitz analysis theory and explain the gap between the analytical and empirical estimates. In this process, we obtain additional information on local stability and robustness of the network, which we exploit in the third part of this subsection where we apply these results to adversarial perturbations. 

\subsubsection{Extending to Discrete Signal Processing}

The AlexNet and the GoogleNet have filters trained on specific datasets, with no closed form parametric description of their weights such as the wavelets. Therefore, instead of using approximations of the continuous signal theory to these networks, we extend the theory to discrete signal processing. We do this by computing the Bessel bounds using the Discrete Fourier Transform along the lines in \cite[Section 4.3]{SZSBEGF13} subsequent to the computation of the operator norms of the discrete linear operators. Given the Bessel bound, the Lipschitz bound is computed using the same estimates derived earlier to the case of continuous signals. 

First we compute the Bessel bounds. 
Note that both networks do not generate feature outputs in hidden layers. Therefore, $B_{m}^{(1)} = B_{m}^{(2)}$, $B_{m}^{(3)} = 0$ for each $1 \leq m \leq M-1$. Now the fourth line of (\ref{eq:lp}) forces $z_m = 0$ for $m = 1, \cdots, M-1$, which makes the second and third lines equivalent. Note that Corollary \ref{thm:prod} is derived by looking at the third and fourth lines of (\ref{eq:lp}). Consequently, the linear program (\ref{eq:lp}) and Corollary \ref{thm:prod} provide the same Lipschitz bounds $L = B_{M}^{(3)} \prod_{m=1}^{M-1} B_m^{(2)}$. 

There are several versions of trained networks for AlexNet and GoogleNet. We consider the MatConvNet \cite{VedL15} pretrained networks that are trained using ImageNet (ILSVRC2012) dataset \cite{ILSVRC15} (the trained networks for both AlexNet and GoogleNet are retrievable at http://www.vlfeat.org/matconvnet/pretrained/). For both pretrained models, there are no cross-channel response normalizations (which appears in the original model \cite{Krizhevsky12}). The features are extracted after the last convolution layer in each network.

We present the Bessel bounds ($B_m^{(2)}$ for $1 \leq m \leq M-1$ and $B_m^{(3)}$ for $m = M$) for each layer of the AlexNet in Table II and the GoogleNet in Table III. Since we are in the discrete case (previous sections discuss signals $f \in L^2(\RR^d)$ and continuous convolutions), we need to adjust the way the Bessel bounds in (\ref{def:b2mnm}) and (\ref{def:b3mnm}) are computed. The adjusted computations for the AlexNet follows \cite[ Section 4.3]{SZSBEGF13}, which uses the Discrete Fourier Transform and takes striding into account. For the GoogleNet, we treat the inception modules (see \cite[Figure 2(b)]{Szegedy15}) as two layers: the first layer is the scattering with the dimension reductions (denoted as ``icpxreduce'' in Table III), and the second layer is the merging after taking convolutions (denoted as ``icpxconv'' in Table III).

\linespread{1.5}

\begin{table}[ht!]
\label{tab:alexnet} 
\centering
 \begin{tabular}{||c | c||} 
 \hline
 Layer & Lip const \\ [0.5ex] 
 \hline\hline
 conv1 & 0.2628  \\ 
 \hline
 conv2 & 6.7761  \\
 \hline
 conv3 & 6.5435  \\
 \hline
 conv4 & 13.3898 \\
 \hline
 conv5 & 16.0937 \\
 \hline
\end{tabular}
\vspace{1mm}
\caption{The Bessel constants (= square root of Bessel bounds) for each layer of AlexNet.}
\end{table}

\begin{table}[ht!]
\label{tab:googlenet} 
\centering
 \begin{tabular}{||c | c|} 
 \hline
 Layer & Lip const \\ [0.5ex] 
 \hline\hline
 conv1 & 5.8608  \\ 
 \hline
 reduce2 & 3.4147  \\
 \hline
 conv2 & 3.0309  \\
 \hline
 icp1reduce & 3.6571 \\
 \hline
 icp1conv & 5.2917 \\
 \hline
 icp2reduce & 3.7994 \\
 \hline 
 icp2conv & 7.6367 \\
 \hline
 \end{tabular}
 \begin{tabular}{|c | c|}
 \hline 
  Layer & Lip const \\ [0.5ex] 
 \hline\hline
 icp3reduce & 2.6642 \\
 \hline 
 icp3conv & 6.0129 \\
 \hline 
 icp4reduce & 2.6403 \\
 \hline 
 icp4conv & 5.1029 \\
 \hline 
 icp5reduce & 2.9825 \\
 \hline 
 icp5conv & 5.5389 \\
 \hline 
 icp6reduce & 3.1758 \\
 \hline
 \end{tabular}
 \begin{tabular}{|c | c||} 
 \hline 
  Layer & Lip const \\ [0.5ex] 
 \hline\hline
 icp6conv & 6.7737 \\
 \hline 
 icp7reduce & 2.2093 \\
 \hline 
 icp7conv & 6.5312 \\
 \hline 
 icp8reduce & 2.2947 \\
 \hline 
 icp8conv & 5.5561 \\
 \hline 
 icp9reduce & 2.8567 \\
 \hline 
 icp9conv & 7.0353 \\
 \hline 
\end{tabular}
\vspace{1mm}
\caption{The Bessel constants (= square root of Bessel bounds) for each layer of GoogleNet.}
\end{table}

\linespread{1}

Using the computed Bessel constants and Corollary \ref{thm:prod}, the estimated Lipschitz constant for the AlexNet is $2.51 \times 10^3$ and for the GoogleNet is $9.67 \times 10^{12}$. Subject to the Bessel computed computed above (which are tight), the CNN Lipschitz bound estimates cannot be improved analytically. Instead we perform an empirical study. Specifically, we randomly take two images $f_1$ and $f_2$ from ImageNet, and compute the ratio $|||\Phi(f_1)-\Phi(f_2)||| / \norm{f_1 - f_2}_2$, where $\Phi$ is the Lipschitz map induced by the network. The empirical Lipschitz constant is the largest ratio among all samples that we take. We sample $10^6$ pairs for this experiment. The resulting empirical constant is $7.32 \times 10^{-3}$ for the AlexNet, and $4.84 \times 10^{-2}$ for the GoogleNet. 

The empirical constants are of significantly smaller order than the analytical constants. In general, two factors explain the gap between the analytical and empirical Lipschitz constant estimates: first, the principal singular vector that optimizes the operator norm in a given layer in not in the range of signals reachable by the previous layer; second, whenever we have ReLU nonlinearity and max pooling, the distance between two vectors tends to shrink. 

The first factor can be partially addressed by considering the norm of tensorial product of all layers instead of considering the product of tensor norms in each layer individually (similar to computing the operator norm of a product of matrices directly instead of upper bounding it by the product of operator norms of each matrix). Both this and the second factor can be addressed by a framework that locally linearizes the network for analysis. We demonstrate how to do this in the following subsection.


\subsubsection{Local Lipschitz Analysis}

To begin with, we estimate the Lipschitz constants without the ReLU functions to illustrate the impact of the nonlinearity. We construct the AlexNet and GoogleNet without the ReLU units by replacing them with the identity functions, and repeat the experiments of taking ratios from pairwise random samples. Empirically, the ratio estimated in this way is $9.08 \times 10^{-2}$ for the AlexNet and $1.10 \times 10^{3}$ for the GoogleNet. Note that these constants are larger than the empirical Lipschitz constants for the networks with the ReLU units.

The nonlinearities also have a non-negligible impact on the Lipschitz constant. To handle them in the analysis, we linearize them locally and compute the local Lipschitz constants. The local Lipschitz constant of $\Phi$ at $f \in \cD$ for $\epsilon$-neighborhood is defined by 
\begin{equation}\label{def:localLip}
L^{\textup{loc}}(f, \epsilon) := \sup_{\substack{f' \in \cD \\ \norm{f'-f}_2 < \epsilon}} \frac{ ||| \Phi(f') - \Phi(f) ||| }{ \norm{f'-f}_2} ~.
\end{equation}
Note that for the case of the AlexNet and the GoogleNet (and similarly for all other discrete networks), the input signal is from a compact domain $\cD = \fI^D$ where $\fI$ is the interval for the pixel values, and $D$ is the dimensionality (the number of pixels). Since $\cD$ is convex, the Lipschitz constant of $\Phi$ is the maximum of the local Lipschitz constants on $\cD$. The rigorous proof of this claim is given in Appendix \ref{appendix:localLip}. 

Using the linearization formulas, we estimate numerically the local Lipschitz constants. The procedure is described as follows. We vectorize \footnote{For a matrix $A = [a_1 | a_2 | \cdots | a_D] \in \RR^{D \times D}$ where $a_1, a_2, \cdots, a_D$ are $D$-dimensional vectors, we vectorize $A$ to be $A^{\textup{vec}} = [a_1^t | a_2^t | \cdots | a_D^t]^t$} the input image and the output feature vector. Also, we use Toepliz matrices $T_1, T_2, \cdots, T_M$ to represent filters in each layer. For any input sample $f$, the CNN generates the output feature vector $\Phi(f)$ by propagating $f$ through $T_m$'s and the nonlinearities that activate only a subset of the pixels for the hidden layer outputs. For the $m$-th layer, we delete (remove) the rows that correspond to the pixels not activated by the ReLU units and max pooling (if they exist) in $T_m$, and the corresponding columns in $T_{m+1}$. In this way we obtain matrices $T'_1, T'_2, \cdots, T'_M$. The product $T'[f] = T'_M T'_{M-1} \cdots T'_2 T'_1$ represents the locally linearized operator for the CNN acting at $f$. For a small $\epsilon$, the local Lipschitz constant at $f$ is estimated by $L^{\textup{loc}}(f, \epsilon) \approx \sigma_{\max}(T'[f])$, the largest singular value of $T'$. The Lipschitz constant for $\Phi$ is thus estimated by 
$$L_c = \max_{f \in \fI^D} L^{\textup{loc}}(f,\epsilon) =
\max_{f \in \fI^D} \sigma_{\max}(T'[f]),$$
where the second equality follows if we take the maximum over the entire compact convex set $\fI^D$ (see Appendix \ref{appendix:localLip}). However, for numerical reasons, we replace $\fI^D$ with a finite number of samples ${\cal F}$ thus obtaining an approximate (lower) bound:
$$ L_c \approx \max_{f\in \cal F } \sigma_{\max}(T'[f]).$$

We follow the procedure described above to estimate the Lipschitz constant for the AlexNet, with ${\cal F}$ having 500 random samples drawn from the ImageNet (ILSVRC2012) dataset. Figure \ref{fig:histAlex} illustrates the histogram of these results. We see that the local Lipschitz constants in our case are between 0.2 and 1.6, hence of order 1. Table \ref{tableIV} summarizes the results of the analytical, empirical and numerical local Lipschitz constants analysis for the AlexNet. 

One would naturally ask if the 500 random samples we chose for this analysis are sufficient to infer an accurate estimate of the Lipschitz constant. 
To address this question we performed two sets of experiments. 
 First we test if the local Lipschitz constant is narrowly distributed over samples in each class and whether the distribution changes for random input signals (i.e. artificial noise input images).  Figure \ref{fig:histTench} depicts the histogram of local Lipschitz constants for images from class ``tench'' (left plot), and compares it with the histogram of local Lipschitz constants for i.i.d. Gaussian noise images (right plot). We note that the local Lipschitz constants for Gaussian noise are much more concentrated around a significantly smaller mean than for the class ``tench''. This implies that the AlexNet behaves differently for different ImageNet samples from the same class. On the other hand the distribution of local Lipschitz constants for images from same class reflects the same range of values as the distribution over all 500 images considered in Figure \ref{fig:histAlex}.
 This experiment gives us confidence that the estimated Lipschitz constant over the 500 ImageNet images is nearly tight.

\begin{figure}[!ht]
\centering
\includegraphics[width=0.5\linewidth]{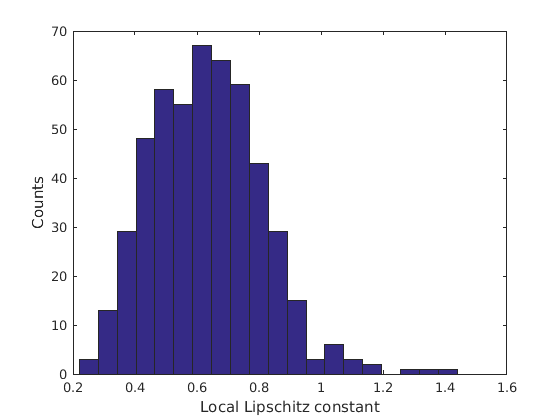}
\caption{The histogram of the local Lipschitz constants for the AlexNet for 500 sample images taken from the ImageNet dataset.}
\label{fig:histAlex}
\end{figure}

\linespread{1.5}

\begin{table}[ht!]
\centering
 \begin{tabular}{||c | c||} 
 \hline
Method & Lip const \\ [0.5ex] 
 \hline\hline
 Analytical estimate: compute Bessel bounds and follow Corollary \ref{thm:prod} & $2.51 \times 10^3$  \\ 
 \hline
 Empirical bound: take quotient from pairs of samples & $7.32 \times 10^{-3}$  \\
 \hline
 Numerical approximation: compute local Lipschitz constants and take the maximum & $1.44$  \\
 \hline
\end{tabular}
\vspace{1mm}
\caption{The Lipschitz constant estimation using three methods for the AlexNet.\label{tableIV}}
\end{table}

\linespread{1}

\begin{figure}[ht!]
\begin{subfigure}{.5\textwidth}
  \centering
  \includegraphics[width=.8\linewidth]{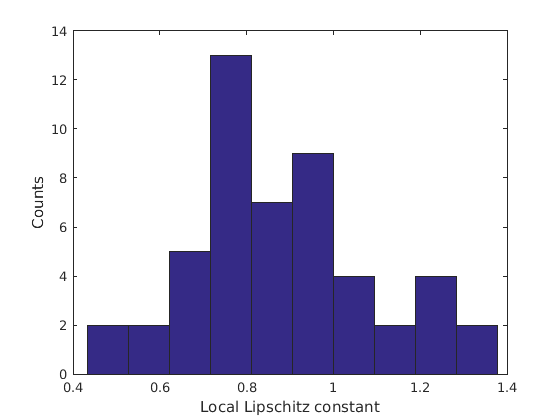}
\end{subfigure}
\begin{subfigure}{.5\textwidth}
  \centering
  \includegraphics[width=.8\linewidth]{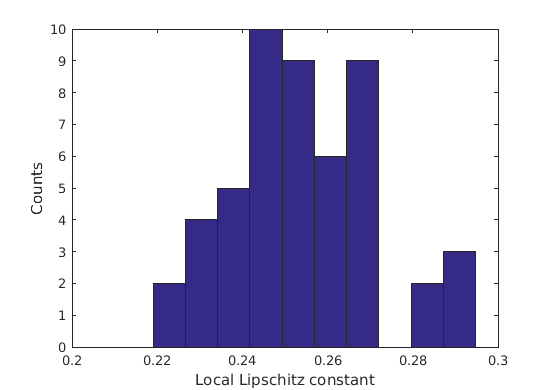}
\end{subfigure}
\caption{Two histograms of local Lipschitz constants for the AlexNet: the left plot contains the results of 50 samples from the class ``tench''; the right plot contains the results from 50 samples from i.i.d. Gaussian distribution of same size ($224\times 224\times 3$).}
\label{fig:histTench}
\end{figure}

On the other hand, as observed from Table IV, the Lipschitz constant computed by taking the maximum of the local Lipschitz constant is about 3 orders of magnitude larger than the empirically computed constant. This surprising observation implies that the direction of maximum variation (the principal singular vector) varies significantly from one ImageNet sample to another. Furthermore, the local Lipschitz constant is large only in a small  neighborhood around each sample. In order to estimate the largest perturbation that achieves the local Lipschitz bound we performed the following experiment. 
For input signal $f$, let $v$ denote the principal singular vector of norm $\norm{v}_2 = 1$ that corresponds to the largest singular value $\sigma_{\max}$. By definition, we have \[\lim_{\epsilon \rightarrow 0} L^{\textup{loc}}(f, \epsilon) = \lim_{t \rightarrow 0} \frac{ |||\Phi(f+t \cdot v) - \Phi(f)||| }{ t} = \sigma_{\max}.\]
Figure \ref{fig:lipVeryLocal} shows how the quotient $|||\Phi(f+h \cdot v) - \Phi(f)||| / h$ changes with $h$. Note that the convergence as $h$ approaches $0$ is very slow. 
In particular, this experiment confirms that the local Lipschitz constant is achievable, hance the numerical estimates in Table \ref{tableIV} are not just numerical artifacts, but actual achievable ratios. On the other hand, Figure \ref{fig:lipVeryLocal} shows that the largest relative variation of the output (i.e. the ratio $ |||\Phi(f) - \Phi(\tilde{f})||| / ||f-\tilde{f}||_2 $) is achieved by small perturbations only. 
In general, given a pair of different image samples from ImageNet, their $l^2$-distance is much larger than $10^{-5}$, so they cannot reflect the local oscillation of $\Phi$.

\begin{figure}[!ht]
\centering
\includegraphics[width=0.5\linewidth]{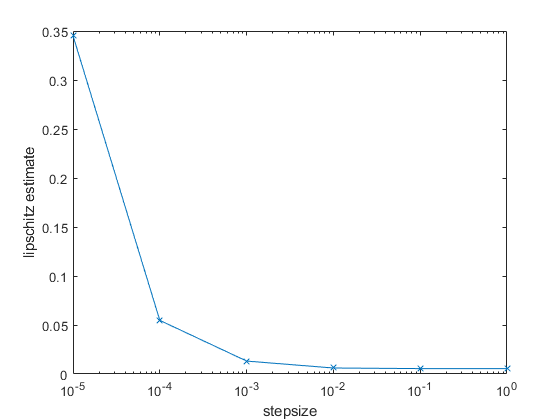}
\caption{The ratio $|||\Phi(f+h \cdot v) - \Phi(f)||| / h$ for different $h$.}
\label{fig:lipVeryLocal}
\end{figure}

\subsubsection{Adversarial Perturbation Induced by the Local Lipschitz Constants}

CNN's such as the AlexNet and the GoogleNet are shown to be vulnerable to small perturbations \cite{SZSBEGF13, Moo17, su2017one}. This kind of instability of those deep networks not only leads to difficulties in cross-model generalization, but also causes serious security problems in practice \cite{papernot2016distillation, zhang2016adversarial}. 
An adversarial perturbation is a small perturbation of the input signal that changes the classification decision of the CNN. The perturbation can be constructed by solving an optimization problem where the wrong classification is considered as a loss in the objective function, as described in \cite{SZSBEGF13}. Various optimization settings can be found in \cite{Moo17, su2017one} where specific restriction on the perturbation is required.

The local Lipschitz analysis carried out in the previous section characterizes the impact of varying the direction of signals perturbations on the output of the CNN. It can be seen that for the same amount of input perturbation, different directions can be chosen to achieve a better adversarial impact on the network performance. We use this observation to create adversarial perturbations below. We show that a relative change of the order of $10^{-2}$ can lead the network to wrongly characterize the input image. 

Since a local Lipschitz constant is associated with a singular vector $v_0$ with $\norm{v_0}_2 = 1$ which is the direction that $\Phi$ varies the most at $f$, we expect this direction gives a perturbation that ``fools'' the CNN more than other directions. The task is to find the smallest $h$ for which $f$ and $f' = f+h \cdot v_0$ are labeled differently by the CNN. We use the AlexNet and empirically search for $h$. For each sample, we find the smallest $h$ that fools the AlexNet. One such example is given in Figure \ref{fig:exPerturb}. We take 50 samples and find that the optimal $h_{\textup{opt}}$'s have order of magnitude $10^3$, which is relatively small compared to $\norm{f}_2$ (we have $227 \times 227 \times 3$ input with pixel values in $[0, 255]$, so the relative change is of the order $10^{-2}$). Note that this order of $h$ is also observed in \cite{Moo17}, where the 2-norm of the perturbation is chosen to be 2000. Further, for each sample, we take 1000 random directions $v_{\textup{rand}}$, and compare the labels given by the AlexNet for $f$ and $f + (h_{\textup{opt}}+\Delta h) \cdot v_{\textup{rand}}$ for a set of different values of $\Delta h$. We plot the percentage of directions that fools the AlexNet on average for these samples in Figure \ref{fig:perFool}. Surprisingly, the direction informed by the local Lipschitz constant performs better than most directions, although at for $h > 10^3$ the quotient $|||\Phi(f+h \cdot v) - \Phi(f)||| / h$ is much smaller than the Lipschitz constant at $f$. Empirically, this implies that the local Lipschitz constant is still important although it decreases fast outside a small region.

\begin{figure}[ht!]
\begin{subfigure}{.33\textwidth}
  \centering
  \includegraphics[width=\linewidth]{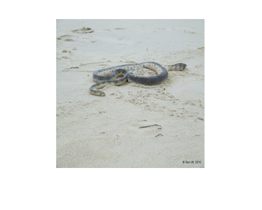}
  \caption{The original image}
\end{subfigure}
\begin{subfigure}{.33\textwidth}
  \centering
  \includegraphics[width=\linewidth]{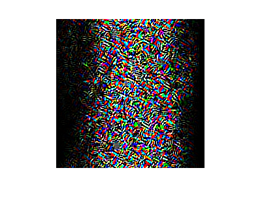}
  \caption{The perturbation (scaled for visibility)}
\end{subfigure}
\begin{subfigure}{.33\textwidth}
  \centering
  \includegraphics[width=\linewidth]{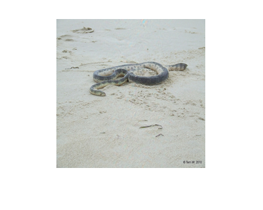}
  \caption{The perturbed image}
\end{subfigure}
\caption{An example of the perturbation along the direction of the singular vector. The left is the original image, the middle is the perturbation which is amplified 1000 times for clear illustration, and the right is the perturbed image. The AlexNet recognizes the original image as ``king snake'' but the perturbed one as ``loggerhead turtle''.}
\label{fig:exPerturb}
\end{figure}

\begin{figure}[!ht]
\centering
\includegraphics[width=0.5\linewidth]{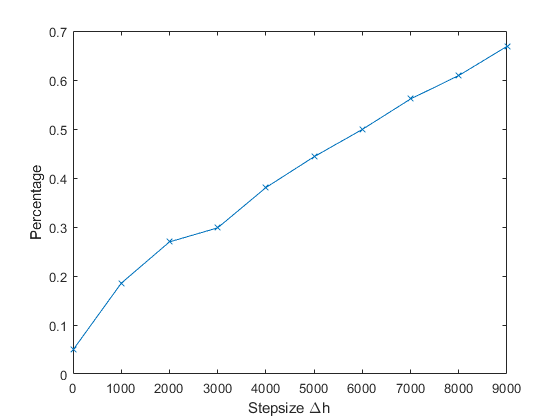}
\caption{Average percentage of successful perturbations in 1000 random directions. $\Delta h = 0$ is the smallest stepsize where the perturbation along the direction informed by the local Lipschitz constants successfully fools the AlexNet.}
\label{fig:perFool}
\end{figure}

\section{Stationary processes}\label{sec:stationary}
Signals (audio or image) are often modeled as random processes. In our case, there are two ways to model the input signal of a CNN: one is to consider $X(t)$ as a random process (field) with some underlying probability space $(\Omega, \fF, \PP)$ with finite second-order moments (see \cite[Chapter 4]{Mallat12}); the other is to regard $X$ as a random variable such that
\[X: (\Omega, \fF, \PP) \rightarrow L^2(\RR^d) ~.\]

We first present the former model for our framework in Section \ref{sec:definecnn}. In the following, we use the notation $X(t)$ to emphasize the time (space) variable $t \in \RR^d$ and $X_t(\omega)$ to emphasize $\omega \in \Omega$. We are interested in studying stationary signals. Fix a realization $X(t) = X_{\omega}(t)$ for some $\omega \in \Omega$. $X(t)$ is said to be strict-sense-stationary (SSS) (see, for instance, \cite{KorS07}, Chapter 16) if all of its finite-order moments are time-invariant (its cumulative distribution does not change with time). The output of a CNN is SSS provided that the input $X$ is SSS. This is stated as the following lemma.
\begin{lemma}\label{lem:sss}
Consider a CNN in the framework of Section \ref{sec:definecnn} in which there is no dilation operation. Let $\Phi$ be the induced Lipschitz continuous map as defined in (\ref{def:Phi}). If $X$ is an SSS process, then so is $\Phi(X)$.
\end{lemma}

\begin{remark}
In general, if we apply dilations for random processes, the signals are no longer stationary after the merge operations. To see a concrete example, let $\theta$ be a random variable taking values uniformly in $[0, 2\pi)$. Consider $X(t) = \cos(t + \theta)$ which has i.i.d. distribution over time and is thus SSS. Note that $Y(t) := X(t) + X(3t) = \cos(t + \theta) +  \cos(3t + \theta) =  2 \cos(2t+\theta) \cos(t)$ has different distributions at $t=0$ and $t=\pi/2$, and is thus not SSS. Therefore, throughout this section, we assume that there is no dilation operation in our CNN.
\end{remark}

Now we state the result that connects the Lipschitz bound derived in Section III with stationary processes.
\begin{theorem}\label{thm:sss}
Consider a CNN in the framework of Section \ref{sec:definecnn} in which there is no dilation operation. Let $X$ and $Y$ be SSS processes with finite second-order moments. Then
\begin{equation}
\EE \left(\Big|\Big|\Big| \Phi(X)-\Phi(Y) \Big|\Big|\Big|^2 \right) \leq L \cdot \EE \left(\abs{X-Y}^2 \right) ~.
\end{equation}
In particular, $||| \Phi(X) |||^2 \leq L \cdot \EE \left(\abs{X}^2 \right)$.
\end{theorem}

The proof parallels that of Section III, and we present it in Appendix D.

As mentioned above, we can also follow the second way to model the signal as a random variable $X: \Omega \rightarrow L^2(\RR^d)$. In this case, we have a random variable with values in a Banach space (see a detailed discussion of such random processes in \cite{LedT91,Led05}). In particular, let $\Phi$ be the map induced by the CNN, and $L_c = \sqrt{L}$ be the Lipschitz constant. Denote $Y = \Phi \circ X$ to be the received random variable. Then by Proposition 1.2 in \cite{Led05}, we have the concentration function $\alpha(L^2(\RR^d),\PP_Y)(r) \leq \alpha(L^2(\RR^d),\PP_X)(r/L_c)$. Suppose $X$ is Gaussian (see \cite{LedT91}, Chapter 2 for the definition in this case) and let $\sigma = \sigma(X) = \sup (\EE \norm{X}_2^2)^{1/2}$. Then similar to the concentration inequality as in Lemma 3.1 in \cite{Led05}, there exists a median $\fM > 0$ for which we have both
\[\PP_Y(\norm{Y-\EE(Y)}_2 \leq \fM) \geq 1/2 \]
and
\[\PP_Y(\norm{Y-\EE(Y)}_2 \leq \fM) \leq 1/2 ~;\]
and we have
\begin{equation}\label{eq:nearlyConcentration}
\PP \left\{ \abs{\norm{Y-\EE(Y)}_2-\fM} > t \right\} \leq \exp \left( -\frac{t^2}{2 \sigma^2 L} \right) ~.
\end{equation}

In signal classification tasks, if we view signals in each class as realizations from a common distribution, then we have the same $\EE(Y)$ for all signals in this class. If the feature $Y$ generated by the CNN is concentrated around $\EE(Y)$, and $\EE(Y)$'s are separated for different classes, then features from different classes will naturally form clusters. Although we do not have exact concentration ($\fM=0$), Inequality (\ref{eq:nearlyConcentration}) demonstrates that $Y$ concentrates in a ``thin'' shell of radius $\fM$ around $\EE (Y)$ provided that we have a small Lipschitz bound $L$. We further promotes making the Lipschitz bound small in designing CNN's in the next section.

\section{Lipschitz bound in classification}\label{sec:classification}

In linear discriminant analysis (LDA) (see, e.g. \cite{XanPT13,Mika99}), it is desired to maximize the ``separation'', or the ``discriminant'', which is the variance between classes divided by the variance within each class (see \cite{Mika99}, Eq (1) and the discussion that follows). We use a similar notion in our (nonlinear discriminant) analysis, albeit its nature of nonlinearity. We define the \emph{discriminant} of two classes $C_1$ and $C_2$ to be
\begin{equation}\label{def:separation}
S = \frac{|||\EE [\Phi(f)|f \in C_1]-\EE[\Phi(f)|f \in C_2]|||^2}{\norm{\Var(\Phi(f)|f \in C_1)}_{*} + \norm{\Var(\Phi(f)|f \in C_2)}_{*}} ~,
\end{equation}
in which $\Phi$ is the nonlinear map induced by the CNN, as defined in (\ref{def:Phi}), $\norm{\cdot}_{*}$ denotes the nuclear norm, and $\Var$ denotes the covariance matrix.

To see how the Lipschitz bound is associated with the separation $S$, we look at the nature of the variance of the output feature $\Phi(f)$. Suppose we have a Gaussian noise $\nu \sim N(0,I)$ and apply a linear transform $A$, then $A \nu$ is also Gaussian with covariance $AA^{t}$. The nuclear norm of its variance is given by \[\norm{\Var(A\nu)}_{*} = \tr AA^{t} = \norm{A}_{\small{\Fr}}^2 ~,\] where $\norm{\cdot}_{\small{\Fr}}$ denotes the Frobenius norm. Since $A$ is linear, its Lipschitz constant is given by $\norm{A}_{\small{\op}}$ and its Lipschitz bound is given by $\norm{A}_{\small{\op}}^2$. Note that the Frobenious norm and the operator norm are equivalent norms, since $\norm{A}_{\op} \leq \norm{A}_{\Fr} \leq \sqrt{n} \norm{A}_{\op}$. 

Motivated by the linear case, we look into replacing $\norm{\Var(\cdot)}_{*}$ in (\ref{def:separation}) with the Lipschitz bound for general CNN's. We consider a CNN with a Gaussian white noise input $\nu \sim N(0, I)$. We assume two classes of signals, $C_1$ and $C_2$ where each class $C_c$ ($c = 1,2$) contains samples from a colored Gaussian noise $\nu_c \sim N(\mu_c, W_c W_c^t)$. We use $L_c$ to denote the Lipschitz bound for the whole system, as illustrated in Figure \ref{fig:lipclass}.

\begin{figure}[!ht]
\centering
\includegraphics[width=0.5\linewidth]{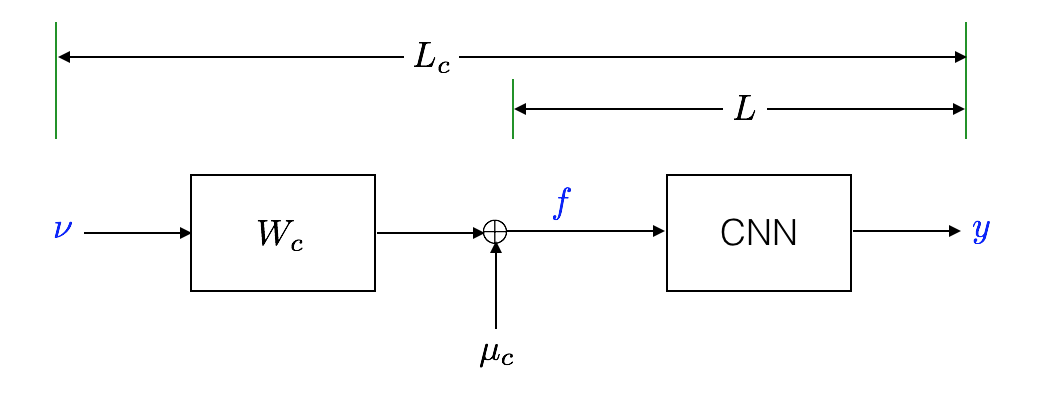}
\caption{Illustration of the Lipschitz bounds $L_c$. Suppose $f$ is an image filtered by $W_c$ (and a bias $\mu_c$) from a white Gaussian noise $\nu \sim N(0, I)$. Then the Lipschitz bound $L_c$ for the class $c$ considers both processes of $W_c$ and the CNN. This bound is not the same for different classes since it depends not only on the CNN but also on $W_c$.}
\label{fig:lipclass}
\end{figure}

We define the \emph{Lipschitz discriminant} to be
\begin{equation}\label{def:separationlip}
\tilde{S} = \frac{|||\EE [\Phi(f)|f \in C_1]-\EE[\Phi(f)|f \in C_2]|||^2}{L_1 + L_2} ~,
\end{equation}
where $L_1$ and $L_2$ are the Lipschitz bounds for Class 1 and Class 2, respectively. 

In Figure \ref{fig:3randn} -- \ref{fig:4rand}, we report the experiments on the discriminative behavior of randomly generated CNN's. We take two classes (number ``3'' and ``8'') of test images from the well-known MNIST database \cite{LecC10}, and randomly build CNN's with three or four convolutional layers and record their discriminant according to (\ref{def:separation}) (plotted on the left-hand-side in each figure) and (\ref{def:separationlip}) (plotted on the right-hand-side in each figure). We then train a linear SVM for each network and plot the error rate of classification against the discriminants. The purpose of this experiment is to show that smaller discriminants lead to better classification results. The reason we use SVM's is to examine the quality of the CNN (feature extractor) given different discriminants, and therefore we choose to train linear SVM's (which works for two classes) with the same regularization parameter. The numerical implementation is done using MATLAB 2016b. We use MatConvNet \cite{VedL15} for constructing the CNN, and the Machine Learning Toolbox in MATLAB for training the SVM's.

\begin{figure}[!ht]
\centering
\includegraphics[width=0.6\linewidth]{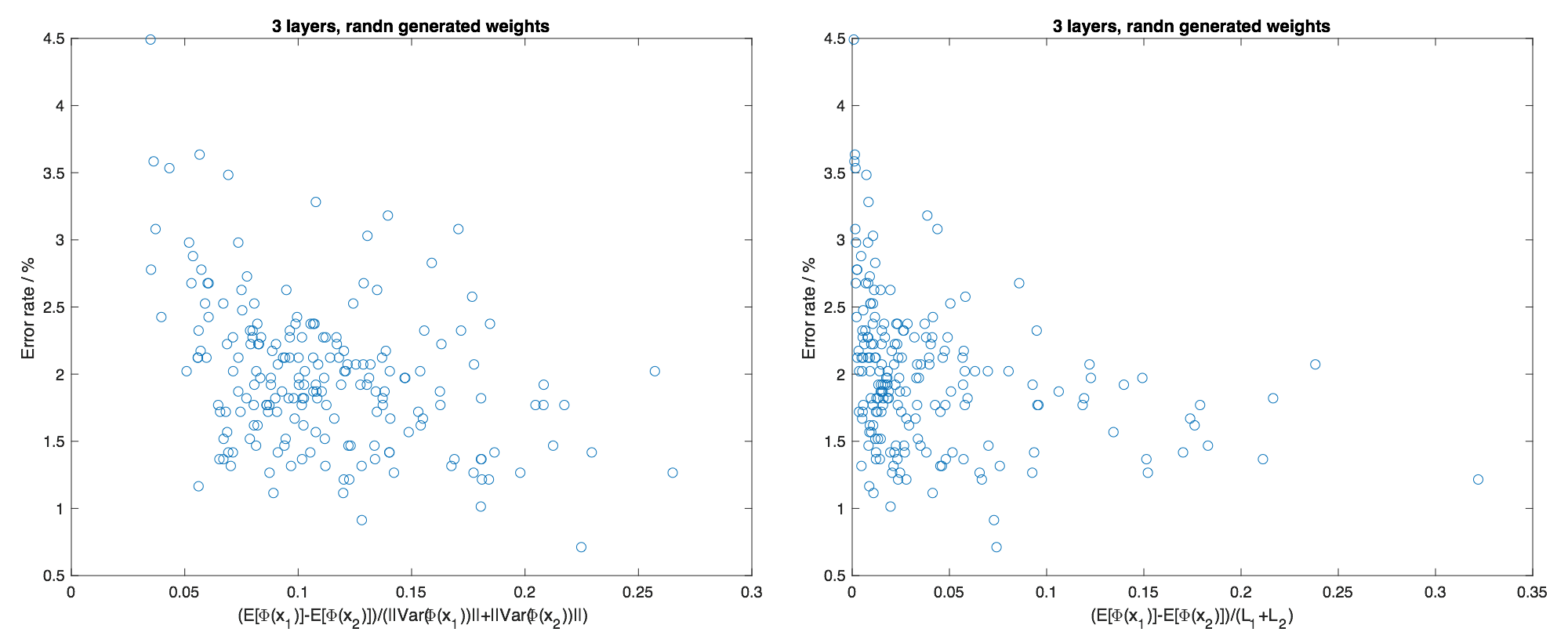}
\caption{Plots of error rate versus discriminant for a three-layer CNN with randomly (normal distributed) generated weights.}
\label{fig:3randn}
\end{figure}

\begin{figure}[!ht]
\centering
\includegraphics[width=0.6\linewidth]{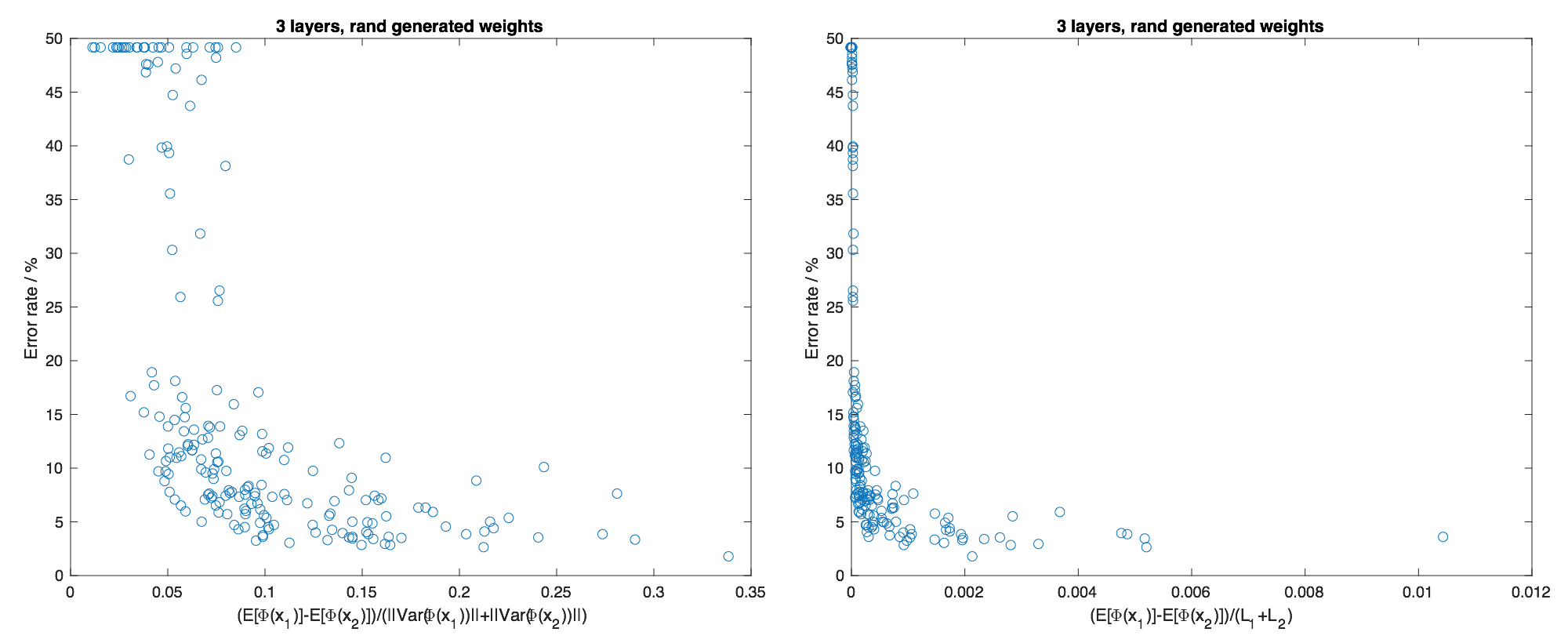}
\caption{Plots of error rate versus discriminant for a three-layer CNN with randomly (uniformly distributed) generated weights.}
\label{fig:3rand}
\end{figure}

\begin{figure}[!ht]
\centering
\includegraphics[width=0.6\linewidth]{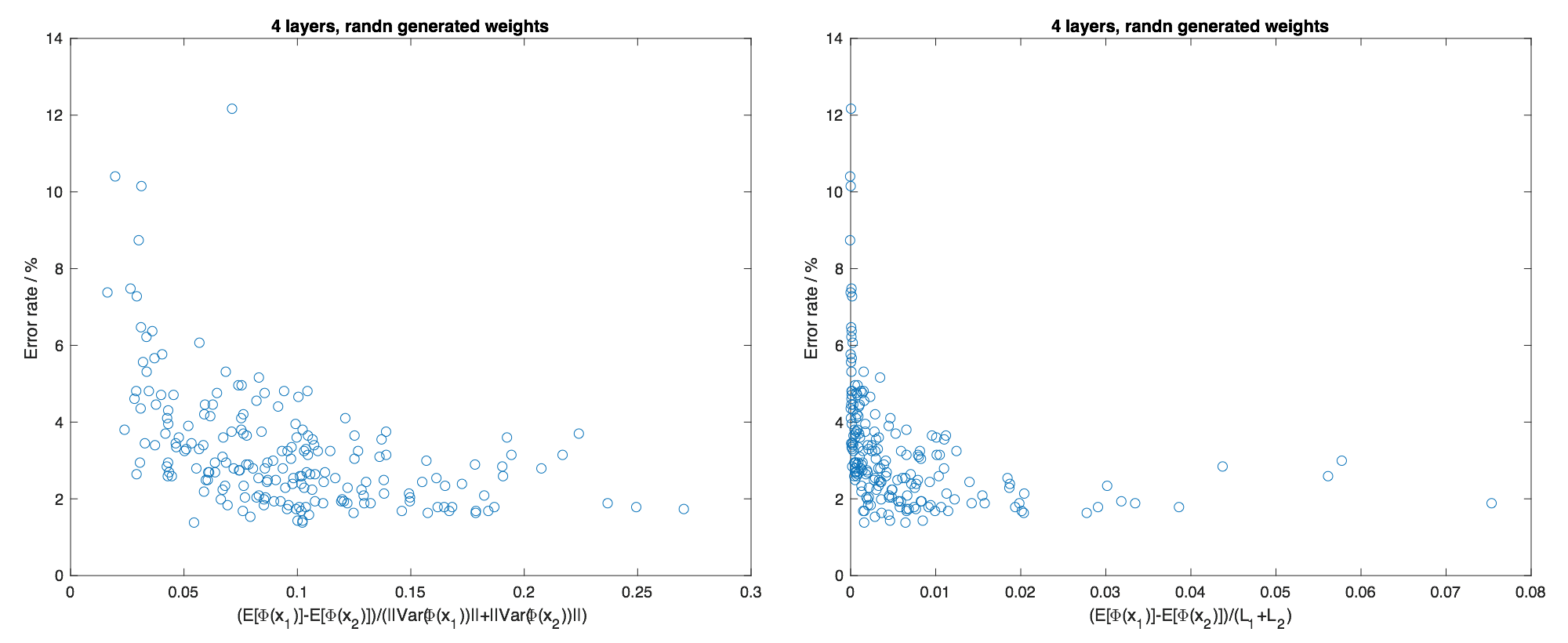}
\caption{Plots of error rate versus discriminant for a four-layer CNN with randomly (normal distributed) generated weights.}
\label{fig:4randn}
\end{figure}

\begin{figure}[!ht]
\centering
\includegraphics[width=0.6\linewidth]{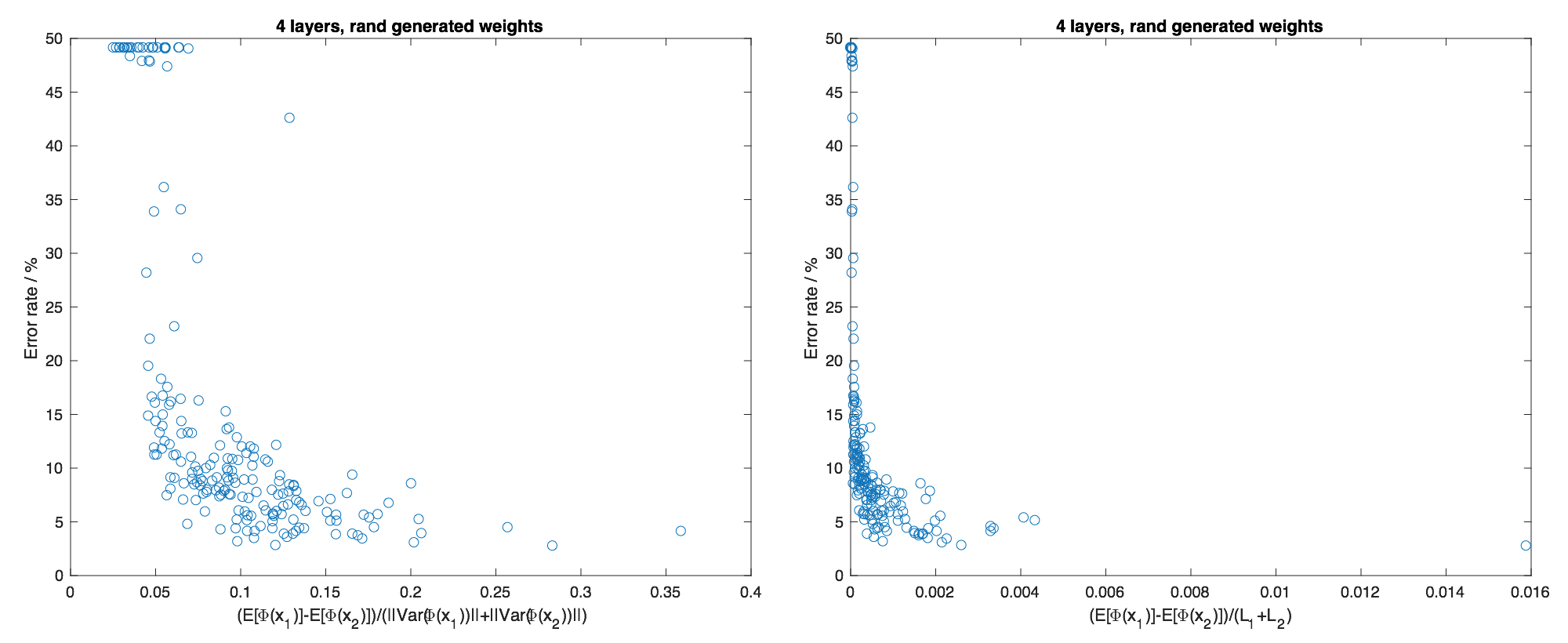}
\caption{Plots of error rate versus discriminant for a four-layer CNN with randomly (uniformly distributed) generated weights.}
\label{fig:4rand}
\end{figure}

As seen from the results, the error rate tends to decrease as the discriminant (\ref{def:separation}) and the Lipschitz discriminant (\ref{def:separationlip}) increase. The trend is clearer when we have more layers. Therefore, either the discriminant or the Lipschitz discriminant is a reasonable penalty term for the training objective function of the CNN. Our analysis in previous chapters can be effectively used to estimate the Lipschitz discriminant for these optimization problems. However, it remains open how to design a training algorithm using the discriminants since the weights appear in both the numerators and denominators in (\ref{def:separation}) and (\ref{def:separationlip}).

\section{Conclusion}\label{sec:conclusion}
In this paper we proposed a general framework for CNN's. We showed that the Lipschitz bound can be calculated by solving a linear program with the Bessel bounds for each layer. We also demonstrated that the Lipschitz bounds play a significant role in the second order statistical description of CNN's. Further, we illustrated that the Lipschitz bounds can be used to form a discriminant that works effectively in classification systems. From the numerical experiments in Section \ref{sec:examples}, we found interesting results that deserve future work. Future works on this topic include mathematically study the distribution of the local Lipschitz constants in a large deviation sense and how the empirical Lipschitz constants depend on the sample distribution. For those questions, the study will not be in a worst-case sense, and thus requires a framework that addresses more randomness.  

\section*{Acknowledgements}
DZ was partially supported by NSF Grant DMS-1413249. RB was partially supported by NSF Grant DMS-1413249, ARO
Grant W911NF-16-1-0008, and LTS Grant H9823031D00560049.  The authors thank the anonymous reviewers for their careful reading of the manuscript and constructive suggestions.

\appendices
\section{Proof of Theorem \ref{thm:lp}}\label{appendix:prooflip}
We are going to show that the optimal value for the linear program (\ref{eq:lp}) is a Lipschitz bound. In particular, we study $\sum_{N \in \cV} \norm{f_N-\tilde{f}_N}_2^2$ as $\sum_{m=1}^M \sum_{N \in \cV_m} \norm{f_N - \tilde{f}_N}_2^2$.

For the $m$-th layer, we mark the signals at the input nodes to be $h_{m,1}, \cdots, h_{m,n_m}$ and the signals at the output nodes to be $h'_{m,1}, \cdots, h'_{m,n'_m}$. We estimate the Lipschitz bound by comparing the output nodes and input nodes for each layer, and then derive a relation between the outputs and the input at the very first layer. Note that with our notation here, $h_{1,1} = f$ and $\tilde{h}_{1,1} = \tilde{f}$.

We first look at the case of no merging. Before we study the input-output relation, note that for the dilation operation illustrated in Figure \ref{fig:dilation}, for two outputs $y_0, \tilde{y}_0$ from inputs $y_1, \tilde{y}_1 \in \RR^d$ respectively, we have
\begin{figure}[!ht]
\centering
\includegraphics[width=0.3\linewidth]{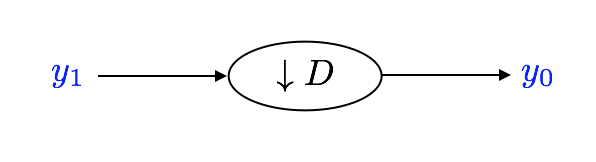}
\caption{The dilation operation. $y_1 \in \RR^d$ is the input and $y_0$ is the output given as $y_0(x) = y_1(Dx)$.}
\label{fig:dilation}
\end{figure}

\begin{equation}\label{eq:dilation}
\begin{aligned}
\norm{y_0-\tilde{y}_0}_2^2 ~=~ & \int \abs{y_1(Dx)-\tilde{y}_1(Dx)}^2 dx \\
~=~ & (\det D)^{-1} \norm{y_1-\tilde{y}_1}_2^2 ~.
\end{aligned}
\end{equation}

Now we look at the illustration in Figure \ref{fig:nomerge}. Since the nonlinearity $\sigma_{m,n'}$ is 1-Lipschitz, and also according to (\ref{eq:dilation}), we have
\[\norm{h'_{m,n'}-\tilde{h}'_{m,n'}}_2^2 \leq (\det D_{m,n'})^{-1} \norm{h^{\spadesuit}_{m,n'}-\tilde{h}^{\spadesuit}_{m,n'}}_2^2 ~.\]
Therefore,
\begin{equation}\label{eq:calBessel}
\begin{aligned}
& \sum_{n'=1}^{n'_m} \norm{h_{m,n'}-\tilde{h}_{m,n'}}_2^2 + \norm{f_{m,n}-\tilde{f}_{m,n}}_2^2 \\
~\leq~ & \sum_{n'=1}^{n'_m} (\det D_{m,n'})^{-1} \norm{h^{\spadesuit}_{m,n'}-\tilde{h}^{\spadesuit}_{m,n'}}_2^2 + \norm{f_{m,n}-\tilde{f}_{m,n}}_2^2 \\
~=~ & \sum_{n'=1}^{n'_m} (\det D_{m,n'})^{-1} \norm{\hat{h}^{\spadesuit}_{m,n'}-\tilde{\hat{h}}^{\spadesuit}_{m,n'}}_2^2 + \norm{\hat{f}_{m,n}-\tilde{\hat{f}}_{m,n}}_2^2 \\
~=~ & \sum_{n'=1}^{m_n'} \int \abs{ \left\{ \begin{bmatrix}
 \D^{(m)} \hat{T}^{(m)}(\omega)\\
 \hat{\Phi}^{(m)}(\omega)
\end{bmatrix} \left(\hat{h}^{(m)}-\tilde{\hat{h}}^{(m)}\right) \right\}_{n'} }^2 dw \\
~\leq~ & \left( \sup_{\omega \in \RR^d} \norm{ \begin{bmatrix}
 \D^{(m)} \hat{T}^{(m)}(\omega)\\
 \hat{\Phi}^{(m)}(\omega)
\end{bmatrix} }_{op}^2 \right) 
\left( \sum_{n=1}^{n_m} \norm{h_{m,n}-\tilde{h}_{m,n}}_2^2 \right) \\
~=~ & B^{(1)}_m \sum_{n=1}^{n_m} \norm{h_{m,n}-\tilde{h}_{m,n}}_2^2 ~,
\end{aligned}
\end{equation}
where in the last two steps, $\hat{h}^{(m)}$ is the column vector whose $n'$-th entry is $\hat{h}_{m,n'}$ (and similarly for $\tilde{\hat{h}}^{(m)}$), and $\{ \cdot \}_{n'}$ denotes the $n'$-th entry of a vector.

In the same manner, we have
\begin{equation*}
\sum_{n'=1}^{n'_m} \norm{h_{m,n'}-\tilde{h}_{m,n'}}_2^2 \leq B^{(2)}_m \sum_{n=1}^{n_m} \norm{h_{m,n}-\tilde{h}_{m,n}}_2^2 ~,
\end{equation*}
and
\begin{equation*}
\norm{f_{m,n}-\tilde{f}_{m,n}}_2^2 \leq B^{(3)}_m \sum_{n=1}^{n_m} \norm{h_{m,n}-\tilde{h}_{m,n}}_2^2 ~.
\end{equation*}

We have completed the analysis of one layer without merging. Now we focus on the merging case, in which the definition of the corresponding Bessel bounds will be clear immediately after we study the three types of merging. Now we look at the relation between the output and input of the merging blocks. 

\begin{figure}[!ht]
\centering
\includegraphics[width=0.3\linewidth]{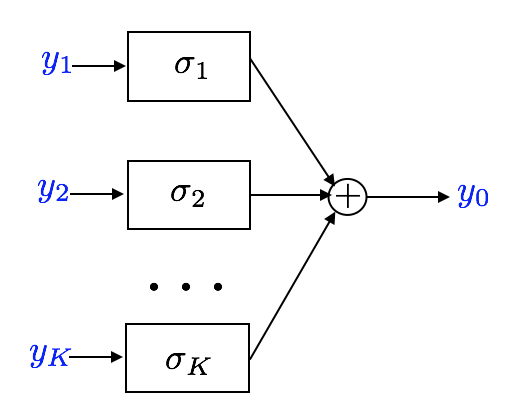}
\caption{Type I merging. $y_0$ is the sum of $\sigma_1(y_1), \cdots, \sigma_K(y_K)$.}
\label{fig:typeI}
\end{figure}

For Type I, as illustrated in Figure \ref{fig:typeI}, we have
\begin{equation}
y_0 = \sum_{k=1}^K \sigma_k (y_k) ~,
\end{equation}
and
\begin{equation}
\tilde{y}_0 = \sum_{k=1}^K \sigma_k (\tilde{y}_k) ~.
\end{equation}
Therefore
\begin{equation}
\begin{aligned}
\norm{y_0 - \tilde{y}_0}_2^2 ~=~ & \norm{\sum_{k=1}^K \sigma_k (y_k)-\sigma_k (\tilde{y}_k)}_2^2 \\
~\leq~ & K \sum_{k=1}^K \norm{\sigma_k (y_k)-\sigma_k (\tilde{y}_k)}_2^2 \\
~\leq~ & K \sum_{k=1}^K \norm{y_k-\tilde{y}_k}_2^2 ~.
\end{aligned}
\end{equation}

\begin{figure}[!ht]
\centering
\includegraphics[width=0.3\linewidth]{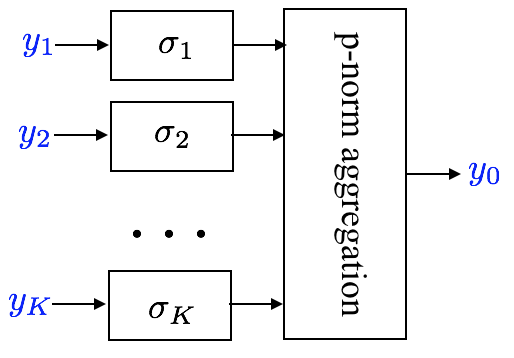}
\caption{Type II merging. $y_0$ is the aggregate of $\sigma_1(y_1), \cdots, \sigma_K(y_K)$ using $p$-norm.}
\label{fig:typeII}
\end{figure}

For Type II, as illustrated in Figure \ref{fig:typeII}, we have
\begin{equation}
y_0 = \left( \sum_{k=1}^K \abs{ \sigma_k ( y_k ) }^p \right)^{1/p},
\end{equation}
and
\begin{equation}
\tilde{y}_0 = \left( \sum_{k=1}^K \abs{ \sigma_k ( \tilde{y}_k ) }^p \right)^{1/p},
\end{equation}
Therefore if $p \leq 2$ we have
\begin{equation*}
\begin{aligned}
& \norm{y_0 - \tilde{y}_0}_2^2 \\
~=~ &  \Bigg\Vert \left( \sum_{k=1}^K \abs{ \sigma_k ( y_k ) }^p \right)^{1/p} - \left( \sum_{k=1}^K \abs{ \sigma_k ( \tilde{y}_k ) }^p \right)^{1/p} \Bigg\Vert_2^2 \\
~\leq~ & \Bigg\Vert \Big( \sum_{k=1}^K | \sigma_k (y_k) - \sigma_k ( \tilde{y}_k ) |^p \Big)^{1/p} \Bigg\Vert_2^2 \\
~\leq~ & K^{2/p-1} \cdot \Bigg\Vert \Big( \sum_{k=1}^K | \sigma_k (y_k) - \sigma_k ( \tilde{y}_k ) |^2 \Big)^{1/2} \Bigg\Vert_2^2 \\
~=~ & K^{2/p-1} \cdot \sum_{k=1}^K \norm{\sigma_k(y_k)-\sigma_k(\tilde{y}_k)}_2^2 \\
~\leq~ & K^{2/p-1} \cdot \sum_{k=1}^K \norm{y_k-\tilde{y}_k}_2^2 ~;
\end{aligned}
\end{equation*}
and if $p > 2$ we have
\begin{equation*}
\begin{aligned}
& \norm{y_0 - \tilde{y}_0}_2^2 \\
~=~ & \Bigg\Vert \left( \sum_{k=1}^K \abs{ \sigma_k ( y_k ) }^p \right)^{1/p} - \left( \sum_{k=1}^K \abs{ \sigma_k ( \tilde{y}_k ) }^p \right)^{1/p} \Bigg\Vert_2^2 \\
~\leq~ & \Bigg\Vert \Big( \sum_{k=1}^K | \sigma_k (y_k) - \sigma_k ( \tilde{y}_k ) |^p \Big)^{1/p} \Bigg\Vert_2^2 \\
~\leq~ & \Bigg\Vert \Big( \sum_{k=1}^K | \sigma_k (y_k) - \sigma_k ( \tilde{y}_k ) |^2 \Big)^{1/2} \Bigg\Vert_2^2 \\
~=~ & \sum_{k=1}^K \norm{\sigma_k(y_k)-\sigma_k(\tilde{y}_k)}_2^2 \\
~\leq~ & \sum_{k=1}^K \norm{y_k-\tilde{y}_k}_2^2 ~.
\end{aligned}
\end{equation*}

For Type III, as illustrated in Figure \ref{fig:typeIII}, we have
$y_0 = \prod_{k=1}^K \sigma_k (y_k)$
and
$\tilde{y}_0 = \prod_{k=1}^K \sigma_k (\tilde{y}_k)$.
Therefore,
\begin{equation*}
\begin{aligned}
& \norm{y_0 - \tilde{y}_0}_2 \\
~=~ & \norm{\prod_{k=1}^K \sigma_k(y_k) - \prod_{k=1}^K \sigma_k(\tilde{y}_k)}_2 \\
~=~ & \Bigg\Vert \prod_{k=1}^K \sigma_k(y_k) + \sum_{J=1}^{K-1} \Big[ - \prod_{k=1}^J \sigma_k(y_k) \prod_{k=J+1}^K \sigma_k(\tilde{y}_k)+ \\
& \qquad \prod_{k=1}^J \sigma_k(y_k) \prod_{k=J+1}^K \sigma_k(\tilde{y}_k) \Big] + \prod_{k=1}^K \sigma_k(\tilde{y}_k) \Bigg\Vert_2 \\
~=~ & \Bigg\Vert \prod_{k=1}^{K-1} \sigma_k(y_k) \cdot (\sigma_K(y_K)-\sigma_K(\tilde{y}_K)) + \sum_{J=2}^{K-1} \prod_{k=1}^{J-1} \sigma_k(y_k) \cdot \\
& \qquad (\sigma_{J}(y_{J})-\sigma_{J}(\tilde{y}_{J})) \cdot \prod_{k=J+1}^K \sigma_k(\tilde{y}_k) + \\
& \qquad (\sigma_1(y_1)-\sigma_1(\tilde{y}_1)) \cdot \prod_{k=2}^K \sigma_k(\tilde{y}_k) \Bigg\Vert_2 \\
~\leq~ & \prod_{k=1}^{K-1} \norm{\sigma_k(y_k)}_{\infty} \cdot \norm{\sigma_K(y_K)-\sigma_K(\tilde{y}_K)}_2 + \\
& \qquad \sum_{J=2}^{K-1} \prod_{k=1}^{J-1} \norm{\sigma_k(y_k)}_{\infty} \cdot \prod_{k=J+1}^K \norm{\sigma_k(\tilde{y}_k)}_{\infty} \cdot \\
& \qquad \norm{\sigma_J(y_J)-\sigma_J(\tilde{y}_J)}_2 + \\
& \qquad \prod_{k=2}^K \norm{\sigma_k(\tilde{y}_k)}_{\infty} \cdot \norm{\sigma_1(y_1)-\sigma_1(\tilde{y}_1)}_2 \\
~\leq~ & \sum_{k=1}^K \norm{\sigma_k(y_k)-\sigma_k(\tilde{y}_k)}_2 \\
~\leq~ & \sum_{k=1}^K \norm{y_k-\tilde{y}_k}_2 ~,
\end{aligned}
\end{equation*}
and thus
\begin{equation}
\norm{y_0-\tilde{y}_0}_2^2 \leq K \sum_{k=1}^K \norm{y_k-\tilde{y}_k}_2^2 ~.
\end{equation}

\begin{figure}[!ht]
\centering
\includegraphics[width=0.3\linewidth]{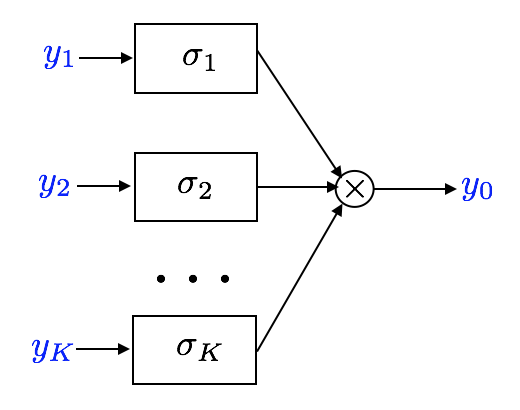}
\caption{Type III merging. $y_0$ is the product of $\sigma_1(y_1), \cdots, \sigma_K(y_K)$. Here $\norm{\sigma_j}_{\infty} \leq 1$ for $j = 1, \cdots, K$.}
\label{fig:typeIII}
\end{figure}


Therefore, when we compare the input nodes and output nodes of the $m$-th layer for the merging case, using the above relations and the definition of $B_m^{(1)}$, we have (see Figure \ref{fig:onelayerdetail})
\begin{equation}
\begin{aligned}
& \sum_1^{n'_m} \norm{h'_{m,n}-\tilde{h}'_{m,n}}_2^2 + \sum_{n=1}^{n_m} \norm{f_{m,n}-f'_{m,n}}_2^2 \\
~\leq~ & B_m^{(1)} \norm{h_{m,n}-\tilde{h}_{m,n}}_2^2 ~.
\end{aligned}
\end{equation}

By the one-one correspondence of the output nodes in the $(m+1)$-th layer and the input nodes in the $m$-th layer, we know that
\begin{equation}
\sum_{n=1}^{n_{m+1}} \norm{h_{m+1,n}-\tilde{h}_{m+1,n}}_2^2 = \sum_{n=1}^{n'_m} \norm{h'_{m,n}-\tilde{h}'_{m,n}}_2^2 ~,
\end{equation}
and therefore,
\begin{equation}
\label{eq:b1}
\begin{aligned}
& \sum_{n=1}^{n_{m+1}} \norm{h_{m+1,n}-\tilde{h}_{m+1,n}}_2^2 + \sum_{n=1}^{n_m} \norm{f_{m,n}-\tilde{f}_{m,n}}_2^2 \\
~\leq~ & B_m^{(1)} \sum_{n=1}^{n_m} \norm{h_{m,n}-\tilde{h}_{m,n}}_2^2 ~,
\end{aligned}
\end{equation}
for $1 \leq m \leq M-1$.

If we do not consider the output generating, then the forward propagation relation is 
\begin{equation}
\label{eq:b2}
\sum_{n=1}^{n_m} \norm{h_{m+1,n}-\tilde{h}_{m+1,n}}_2^2 \leq B_m^{(2)} \sum_{n=1}^{n_m} \norm{h_{m,n}-\tilde{h}_{m,n}}_2^2 ~,
\end{equation}
for $1 \leq m \leq M-1$,
and similarly, considering the output generating nodes alone gives
\begin{equation}
\label{eq:b3}
\sum_{n=1}^{n_m} \norm{f_{m,n}-\tilde{f}_{m,n}}_2^2 \leq B_m^{(3)} \sum_{n=1}^{n_m} \norm{h_{m,n}-\tilde{h}_{m,n}}_2^2 ~,
\end{equation}
for $1 \leq m \leq M$.

Since we would like to compare $\sum_{m=1}^M \sum_{n=1}^{n_m} \norm{f_{m,n}-\tilde{f}_{m,n}}_2^2$ with $\norm{h_{1,1}-\tilde{h}_{1,1}}_2^2$, by (\ref{eq:b1})-(\ref{eq:b3}), we see that the maximal value of the linear program (\ref{eq:lp}) gives a Lipschitz bound. 


\section{Proof of Corollary \ref{thm:prod}}\label{appendix:prooflipcor}

From the definitions of $B_{m,n}^{(1)}$, $B_{m,n}^{(2)}$ and $B_{m,n}^{(3)}$ (\ref{def:b1mn})-(\ref{def:b3mn}) it is obvious that 
\begin{equation}
B_{m,n}^{(1)} \leq B_{m,n}^{(2)} + B_{m,n}^{(3)}
\end{equation}
and from (\ref{def:b1m})-(\ref{def:b3m}), as well as (\ref{def:b1mnm})-(\ref{def:b3mnm}), we have hence
\begin{equation}
B_m^{(1)} \leq B_m^{(2)} + B_m^{(3)}
\end{equation}
for each $m$. Then note that if $\{y_m\}_{m=0}^{M-1}$ and $\{z_m\}_{m=0}^{M-1}$ are the maximums of the linear program (\ref{eq:lp}), then
\begin{equation}
z_m \leq B_{m}^{(1)} y_{m-1} - y_m, \qquad 1 \leq m \leq M-1,
\end{equation}
and
\begin{equation}
z_M \leq B_{M}^{(1)} y_{M-1}
\end{equation}
(note that $B_M^{(1)} = B_M^{(3)}$).

We take the sum over all $m$'s to get (denote $y_M = 0$)
\begin{equation}
\begin{aligned}
\sum_{m=1}^M z_m 
~\leq~ & \sum_{m=1}^M B_m^{(1)} y_{m-1} - y_m \\
~=~ & \sum_{m=0}^{M-1} B_{m+1}^{(1)} y_m - \sum_{m=1}^{M-1} y_m \\
~=~ & B_1^{(1)} + \sum_{m=1}^{M-1} (B_{m+1}^{(1)}-1) y_m ~.
\end{aligned}
\end{equation}
Also, $y_m \leq B_m^{(2)} y_{m-1}$ implies $y_m \leq B_m^{(1)} y_{m-1}$, so
\begin{equation}
\begin{aligned}
\sum_{m=1}^M z_m
~\leq~ & B_1^{(1)} + \sum_{m=1}^{M-1} (\max\{1,B_{m+1}^{(1)}\} - 1) \cdot \\
& \prod_{m'=1}^m \max\{1,B_{m'}^{(1)}\} \\
~=~ & \prod_{m=1}^M \max\{1,B_m^{(1)}\} ~.
\end{aligned}
\end{equation}

\section{The Banach Algebra (\ref{def:balg})}\label{appendix:banach}

We first show that we indeed have a Banach algebra in (\ref{def:balg}).
\begin{lemma}
$\cB$ as defined in (\ref{def:balg}) is a Banach algebra, where the $+$ operation is pointwise addition, and the $\cdot$ operation is the convolution defined by
\begin{equation}\label{def:convbalg}
f \ast g = \left( \hat{f} \hat{g} \right)^{\mathsf{v}} ~,
\end{equation}
where ``$~^\mathsf{v}~$'' denotes the inverse Fourier transform. 
\end{lemma}
\begin{proof}
Note that $\cB$ is closed under the convolution in the sense of (\ref{def:convbalg}) because $\hat{f} \hat{g} \in L^{\infty}(\RR^d)$ and therefore is also in $\cS'(\RR^d)$.
Since the Fourier transform is an isomorphism on $\cS'(\RR^d)$, the inverse Fourier transform of $\hat{f} \hat{g}$ also lies in $\cS'(\RR^d)$.

After the closedness is clear, it is trivial to check that $\cB$ is indeed an algebra. The fact that $\cB$ is a Banach algebra is due to the norm inequality
\begin{equation}
\norm{\hat{f}\hat{g}}_{\infty} \leq \norm{\hat{f}}_{\infty} \Big\Vert \hat{g} \Big\Vert_{\infty} ~.
\end{equation}
\end{proof}

The definition of the Banach Algebra becomes natural after the Bessel bounds (\ref{def:b1mn})-(\ref{def:b3mn}) are defined. Of course, in practice we can consider only filters lie in the space $L^1(\RR^d)$. The Banach Algebra (\ref{def:balg}) is a larger space, and it also has some practical consideration. Suppose we have a network where there is aggregation of two layers, then we notice that this does not fall in our general model. Nevertheless, we can add several layers of $\delta$-function, to make it fall in our framework. This is illustrated in Figure \ref{fig:equiv}.

In the definition (\ref{def:balg}), the $L^{\infty}$ norm is considered in the usual sense, that is, we only consider $\hat{f}$ to be a well-defined ordinary function in $L^{\infty}(\RR^d)$. Then the convolution operation should be understood as $f \ast g = (\hat{f} \cdot \hat{g})^{\mathsf{v}}$. Then obviously the Banach Algebra $\cB$ is closed and well-defined under the convolution operation.

Under this definition, if we don't choose a smooth (in the frequency domain) filter, then in the signal domain we do not have good decay and it is possible to have infinite $L^1$ norm. Even if we choose signals whose Fourier transform is in $C_c^{\infty}(\RR^d)$, we have a coarse approximation by using Young's inequality. Details can be seen in the example given in \cite{BSZ17}.

\begin{figure}[!ht]
\centering
\includegraphics[width=0.5\linewidth]{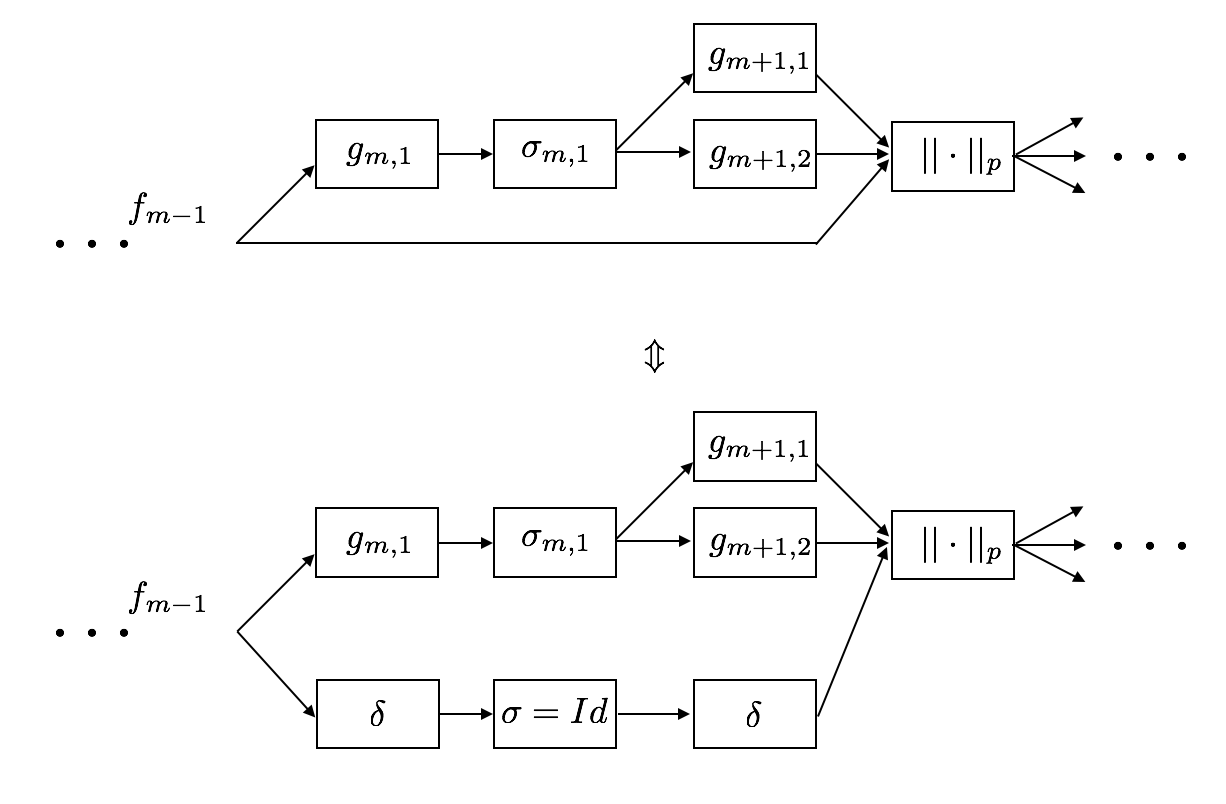}
\caption{Use $\delta$ function to equivalently represent a CNN.}
\label{fig:equiv}
\end{figure}

\section{Lipschitz constants and local Lipschitz constants}\label{appendix:localLip}
For CNN's such as the AlexNet and the GoogleNet, the Lipschitz constant is the maximum among all the local Lipschitz constants (see Section \ref{subsec:alexgoogle}). In particular, we have the following result.

\begin{proposition}
Let $\Phi: \cD \rightarrow \cR$ be a Lipschitz continuous function on a compact convex domain $\cD \in \RR^D$ with the Lipschitz constant
\[ L_c := \max_{\substack{f,g \in \cD \\ f \neq g}} \frac{|||\Phi(f) - \Phi(g)|||}{\norm{f - g}_2} ~,\]
where $|||\cdot|||$ is a well-defined norm on $\cR$. Suppose the local Lipschitz constant at $f \in \cD$ for some $\epsilon > 0$ is $L^{\textup{loc}}(f, \epsilon)$ as defined in (\ref{def:localLip}). Then
\begin{equation}\label{eq:globalIsTheMaxOfAllLocal}
L_c = \max_{f \in \cD} L^{\textup{loc}}(f, \epsilon) ~.
\end{equation}
\end{proposition}

\begin{proof}
Assume on the contrary that (\ref{eq:globalIsTheMaxOfAllLocal}) is not true. Then $L_c >  \max_{f \in \cD} L^{\textup{loc}}(f, \epsilon)$. Suppose $L_c = 2 \delta  + \max_{f \in \cD} L^{\textup{loc}}(f, \epsilon)$. Then there exists $f, g \in \cD$ for which
\begin{equation}\label{assumptionCuoDe}
\frac{|||\Phi(f) - \Phi(g)|||}{\norm{f - g}_2} > \delta  + \max_{f \in \cD} L^{\textup{loc}}(f, \epsilon) ~.
\end{equation}
Let $I = \{ h ~|~ h = (1-t)f + tg , 0 \leq t \leq 1 \} \subset \cD$ be the line segment that joins $f$ and $g$. Take \[I' = \{ h ~|~ h = (1-t)f + tg ,~ t = 0, \frac{\epsilon}{2} , \epsilon, \frac{3 \epsilon}{2} , \cdots , \frac{\epsilon}{2}\left \lfloor \frac{2}{\epsilon} \right \rfloor, 1 \}~.\] Let $N = \abs{I'}$ denote the number of elements in $I'$. Let $h_n = (1-\frac{\epsilon}{2})f + \frac{\epsilon}{2}g$ for $n = 1, \cdots, N-1$ and $h_N = g$. Then since $\norm{h_n - h_{n+1}}_2 \leq \epsilon $,we have
\[
|||\Phi(h_n) - \Phi(h_{n+1})||| \leq L^{\textup{loc}}(h_n, \epsilon) \cdot \norm{h_n - h_{n+1}}_2, \qquad n = 1, 2, \cdots, N-1 ~.
\]
But $L^{\textup{loc}}(h_n, \epsilon) \leq \max_{f \in \cD} L^{\textup{loc}}(f, \epsilon)$, so we have
\[
|||\Phi(h_n) - \Phi(h_{n+1})||| \leq \max_{f \in \cD} L^{\textup{loc}}(f, \epsilon) \cdot \norm{h_n - h_{n+1}}_2, \qquad n = 1, 2, \cdots, N-1 ~.
\]
Summing over $n = 1, 2, \cdots, N-1$ and applying the triangle inequality for norms, we have 
\[
|||\Phi(f) - \Phi(g)||| \leq \sum_{n=1}^{N-1}  |||\Phi(h_n) - \Phi(h_{n+1})|||  \leq \max_{f \in \cD} L^{\textup{loc}}(f, \epsilon) \cdot \sum_{n=1}^{N-1} \norm{h_n - h_{n+1}}_2 = \max_{f \in \cD} L^{\textup{loc}}(f, \epsilon) \norm{f-g}_2 ~,
\]
where the last equality come from the fact that $h_n$'s are all on the same line. But this implies \[  \frac{|||\Phi(f) - \Phi(g)|||}{\norm{f - g}_2} \leq \max_{f \in \cD} L^{\textup{loc}}(f, \epsilon) ~, \]
which contradicts (\ref{assumptionCuoDe}). Therefore the assumption cannot be true and we conclude with (\ref{eq:globalIsTheMaxOfAllLocal}).
\end{proof}

\section{Proof of Lemma \ref{lem:sss}}\label{appendix:lemsss}
The proof of Lemma \ref{lem:sss} lies on the following two facts.\\
1) If $X$ is SSS, then $\sigma(X(t))$, where $\sigma$ is a pointwise function, is also SSS;\\
2) If $X$ is SSS, then $X \ast g(t)$ defined as
\begin{equation}
(X \ast g)_{\omega}(t) = \int X_{\omega}(t-s)g(s) ds ~,
\end{equation}
is also SSS.
To see 1), we need to show
\begin{equation}
\begin{aligned}
& \PP \Big\{ \sigma(X_{t_1+\tau}) \in A_1, \cdots, \sigma(X_{t_n+\tau}) \in A_n \Big\} \\
~=~ & \PP \Big\{ \sigma(X_{t_1}) \in A_1, \cdots, \sigma(X_{t_n}) \in A_n \Big\}
\end{aligned}
\end{equation}
for any $t_1, \cdots, t_n, \tau \in \RR^d$ and any $A_1, \cdots, A_n \in \fF$.
Let $B_j = \sigma^{-1} (A_j) = \{c \in \CC: \sigma(c) \in A_j \}$ for $j = 1, \cdots, n$. The above equality reads
\begin{equation}
\begin{aligned}
& \PP \Big\{ X_{t_1+\tau} \in B_1, \cdots, X_{t_n+\tau} \in B_n \Big\} \\
~=~ & \PP \Big\{ X_{t_1} \in B_1, \cdots, X_{t_n} \in B_n \Big\} ~,
\end{aligned}
\end{equation}
which holds true due to the assumption that $X$ is SSS.\\
To see 2, note that since $X$ is SSS there exists a semigroup of measure-preserving transformation
\[ \left\{ T^t: \Omega \rightarrow \Omega \right\}_{t \in \RR^d} \]
associated with $X$ such that
\[T^s T^t = T^{s+t}\]
for each $s, t \in \RR^d$; and a function $f$ such that
\begin{equation}
f(T^t \omega) = X_t(\omega) ~,
\end{equation}
for each $\omega \in \Omega$, $t \in \RR^d$.
Thus
\begin{equation}
X \ast g(t) = \int f \left( T^{t-s} \omega \right) g(s) ds ~.
\end{equation}
For any $t_1, \cdots, t_n \in \RR^d$, $A_1, \cdots, A_n \in \fF$, let
\begin{equation}
\tilde{\Omega}_{\tau} = \left\{ \omega \in \Omega: (X \ast g)_{t_1+\tau}(\omega) \in A_1, \cdots, (X \ast g)_{t_n+\tau}(\omega) \in A_n \right\} ~.
\end{equation}
For $\omega \in \tilde{\Omega}_{\tau}$, note that $T^{\tau} \omega$ satisfies
\[ (X \ast g)_{t_1}(\omega) \in A_1, \cdots, (X \ast g)_{t_n}(\omega) \in A_n ~.\] 
Since $T^{\tau}$ is measure-preserving, we have $\PP(\tilde{\Omega}_{\tau}) = \PP(\tilde{\Omega}_0)$. Thus $X \ast g$ is SSS.

Given the two facts and that there is no dilation, Lemma \ref{lem:sss} is proved by tracking from the input to each output of the CNN.

\section{Proof of Theorem \ref{thm:sss}}\label{appendix:stationary}
Since the input $X$ and $Y$ are SSS, so are the signals at all input and output nodes of the CNN. Therefore we can apply the Wiener-Khinchin Theorem to relate the auto-correlation with the power spectrum.

Consider an SSS process $Z$ that are filtered by some fixed $g \in \cB$. Denote $W = Z \ast g$. Then we have $R_W(0) = \int \hat{S}_W(\omega) d\omega$. Note that we have the transfer relation
\begin{equation}
\hat{S}_W (\omega) = \hat{S}_Z (\omega) \cdot \abs{\hat{g} (\omega)}^2 ~.
\end{equation}
That is to say,
\begin{equation}
\EE \left( \abs{W}^2 \right) = \int \hat{R}_W (\omega) \abs{\hat{g}(\omega)}^2 d\omega ~.
\end{equation}
More generally, due to linearity of $\EE$, if we have two inputs $Z$ and $\tilde{Z}$ and a family of filters $\{g_j\}_{j \in J}$, we have
\begin{equation}
\begin{aligned}
&\EE \left( \sum_j \Big| Z \ast g_j  - \tilde{Z} \ast g_j \Big|^2 \right) \\
~=~ & \sum_j \int \hat{S}_{Z-\tilde{Z}}(\omega)  \abs{\hat{g}_j (\omega)}^2 d\omega \\
~=~ & \int \hat{S}_{Z-\tilde{Z}} (\omega) \sum_j \abs{\hat{g}_j}^2 (\omega) d\omega \\
~\leq~ & \int \hat{S}_{Z-\tilde{Z}} (\omega) d\omega \cdot \norm{ \sum_j \abs{\hat{g}_j}^2 }_{\infty} \\
~=~ & \EE \left( \abs{Z-\tilde{Z}}^2 \right) \cdot \norm{ \sum_j \abs{\hat{g}_j}^2 }_{\infty} ~.
\end{aligned}
\end{equation}
With this, we can compare the correlation on the first input nodes with the outputs of the CNN similar to what we did in the proof of Theorem \ref{thm:lp}. Note that for merging, the inequalities still hold when $\norm{\cdot}_2^2$ are replaced with $\EE \abs{\cdot}^2$. 

\ifCLASSOPTIONcaptionsoff
  \newpage
\fi



\bibliographystyle{IEEEtran}
\bibliography{convref}

\end{document}